\documentclass[acmsmall,screen,nonacm]{acmart}
\pdfoutput=1 

\usepackage{tabularx}  
\usepackage{array}     
\usepackage[T1]{fontenc}
\usepackage{amsfonts}
\usepackage{stmaryrd}
\usepackage{amsmath}
\usepackage{bm}
\usepackage{bbold}
\usepackage{bbm}
\usepackage{mathtools}
\usepackage{mathpartir}
\usepackage{bbding}
\usepackage{color}
\usepackage[ruled,noend]{algorithm2e}
\usepackage[capitalize,nameinlink]{cleveref}
\usepackage{xparse}
\usepackage{ifthen}
\usepackage{thmtools}
\usepackage{thm-restate}

\usepackage[scr=boondox]{mathalpha}

\usepackage{subcaption}
\usepackage{wrapfig}
\usepackage{tikz}
\usepackage{newfile}
\usepackage{totcount}
\newtheorem{theorem}{Theorem}
\newtheorem{claim}{Claim}
\newtheorem{lemma}{Lemma}
\newtheorem{definition}{Definition}
\definecolor{blue-green}{rgb}{0,0.5,0.5}

\regtotcounter{@todonotes@numberoftodonotes}

\usepackage{listings}

  \lstdefinestyle{tinyc}{
    basicstyle=\scriptsize\ttfamily,
    keywordstyle=\color{blue}
  }

  \lstdefinestyle{normalc}{
    basicstyle=\ttfamily,
    numbers=none,
    keywordstyle=\color{blue}
  }

  \lstdefinestyle{inlinec}{
    basicstyle=\ttfamily
  }

\usetikzlibrary{backgrounds}
\tikzstyle{every picture}+=[remember picture]

\usepackage{todonotes}

\newcommand{\RM}[1]{\todo[inline]{\textcolor{red}{\textbf{RM}: #1}}}
\newcommand{\VRS}[1]{\todo[inline]{\textcolor{blue}{\textbf{Sathiya}: #1}}}

\newcommand{\Sadegh}{\textcolor{purple}}

\usepackage{tikz}
\usetikzlibrary{automata}
\tikzset{gadget/.style={->,>=stealth,initial text=,minimum size=7pt,auto,on grid,scale=1,inner sep=1pt,node distance=1cm}}
\tikzset{every state/.style={minimum size=15pt,inner sep=1pt,fill=black!10,draw=black!70,thick}}

\newbool{inappendix}
\appto\appendix{\booltrue{inappendix}}

\usetikzlibrary{backgrounds,calc}

 \usetikzlibrary{ shapes, fit, arrows}
 \usetikzlibrary{decorations} 
\DeclareMathSymbol{\mdot}{\mathord}{symbols}{"01}
\usepackage{enumerate}
\usepackage{paralist}
\usepackage{lscape}
\usepackage{float}
\usepackage{microtype}
\usepackage{enumitem}
\usepackage{tikzit}

\tikzstyle{plain circle}=[fill=white, draw=black, shape=circle]
\tikzstyle{black circle}=[fill=black, draw=black, shape=circle]
\tikzstyle{blue circle}=[fill=blue, draw=black, shape=circle]
\tikzstyle{red circle}=[fill=red, draw=black, shape=circle]
\tikzstyle{blue-green circle}=[fill={rgb,255: red,0; green,128; blue,128}, draw=black, shape=circle]

\tikzstyle{direction edge}=[->]
\tikzstyle{dashed directional edge}=[->, dashed]
\tikzstyle{red full edge}=[draw=red, ->]
\tikzstyle{dashed simple edge}=[-, dashed]
\tikzstyle{double directional}=[double, ->]
\tikzstyle{dash double}=[dashed, double, ->]
\tikzstyle{new edge style 0}=[-, dashed, draw=red]
\tikzstyle{blue}=[-, draw=blue]
\tikzstyle{blue directiona}=[draw=blue, ->]
\tikzstyle{blue dashed}=[dashed, draw=blue, ->]

\definecolor{blue-green}{rgb}{0,0.5,0.5}

\def\nats{\mathbb{N}}

\def\reals{\mathbb{R}}
\def\preals{\mathbb{R}_{+}}
\def\nnreals{\mathbb{R}_{\geq 0}}
\def\tuple#1{{( #1 ) }}
\def\set#1{{\{ #1 \}}}

\def\calQ{\mathcal{Q}}
\def\cK{\mathcal{K}}

\def\calF{\mathcal{F}}
\def\calB{\mathcal{B}}
\def\calT{\mathcal{T}}
\def\calD{\mathcal{D}}

\def\calS{\mathcal{S}}
\def\calC{\mathcal{C}}
\def\calQ{\mathcal{Q}}

\def\mbf#1{\mathbf{#1}}

\def\bfy{\mbf{y}}

\def\bfv{\mbf{v}}
\def\bfe{\mbf{e}}

\def\sfc{\msf{c}}

\def\bbE{\mathbb{E}}

\def\frakc{\mathfrak{c}}
\def\frakC{\mathfrak{C}}

\def\pp{\mathbb{P}}

\newcommand{\msf}[1]{\mathsf{#1}}
\newcommand{\mtt}[1]{{\small{\mathtt{#1}}}}
\newcommand{\ttt}[1]{{\small{\texttt{#1}}}}

\DeclareDocumentCommand{\autstep}{O{}}{%
        \xrightarrow{#1}%
        }
\DeclareDocumentCommand{\langof}{O{} m}{%
  \mathsf{L}_{#1}(#2)%
  }

\newcommand{\pGCL}{\mathsf{pGCL}}
\newcommand{\RTAS}{\mathsf{RTAS}}

\newcommand{\pChoice}[1]{\oplus_{#1}}

\newcommand{\codeStyleMath}[1]{\mathtt{#1}}

\newcommand{\ttprog}{\codeStyleMath{prog}}
\newcommand{\ttexec}{\codeStyleMath{exec}}

\newcommand{\Inv}{\mathsf{Inv}}

\newcommand{\bbP}{\mathbb{P}}

\newlist{inlinelist}{enumerate*}{1}
\setlist*[inlinelist,1]{%
  label=(\alph*),
}

\newcommand{\setOfNaturals}{\mathbb{N}}

\newcommand{\listingsInLatex}{Code}

\newcommand{\codeRef}[1]{System \ref{#1}} 
\newcommand{\smallCodeRef}[1]{Sys. \ref{#1}} 

\DeclareCaptionType{system}[Sys.][List of Systems]
\crefname{system}{System}{Systems}
\Crefname{system}{System}{Systems}

\newcommand{\node}{\mathsf{node}}

\makeatletter
\newtheoremstyle{dazzle}%
{.5\baselineskip\@plus.2\baselineskip
  \@minus.2\baselineskip}
{.5\baselineskip\@plus.2\baselineskip
  \@minus.2\baselineskip}
{\@acmplainbodyfont}
{\@acmplainindent}
{\bfseries}
{.}
{.5em}
{\thmname{\textcolor{red}{\textbf{#1}}}\thmnumber{ \textcolor{red}{\textbf{#2}}}\thmnote{ {\@acmplainnotefont(\textcolor{blue}{#3})}}}
\makeatother

\usepackage[framemethod=TikZ]{mdframed} 
\newcounter{proofrule}[section]
\renewcommand{\theproofrule}{{\ifinappendix\Alph{section}\else\arabic{section}\fi}.\arabic{proofrule}}
\newenvironment{proofrule}[1][]{%
\refstepcounter{proofrule}%
\smallskip
\ifstrempty{#1}%
{\mdfsetup{%
frametitle={%
\tikz[baseline=(current bounding box.east),outer sep=0pt]
\node[anchor=east,rectangle,fill=blue!20]
{\strut Proof Rule~\theproofrule};}}
}%
{\mdfsetup{%
frametitle={%
\tikz[baseline=(current bounding box.east),outer sep=0pt]
\node[anchor=east,rectangle,fill=blue!20]
{\strut\sffamily \textsc{Proof Rule~\theproofrule:}~#1};}}%
}%
\mdfsetup{innertopmargin=3pt,linecolor=blue!20,%
skipbelow=7pt, skipabove=4pt%
linewidth=2pt,topline=true,%
outerlinewidth=1pt,
}
\begin{mdframed}[]\relax%
}{\end{mdframed}}

\crefname{proofrule}{Rule}{Rules}
\Crefname{proofrule}{Rule}{Rules}

\newcommand{\sfskip}{\mathsf{skip}}
\newcommand{\Vars}{\mathsf{GVars}}
\newcommand{\LVars}{\mathsf{LVars}}
\newcommand{\lVars}{\mathsf{LVars}}

\newcommand{\deslang}{{\mathsf{DES}}}
\newcommand{\toolname}{{\mathsf{Leonides}}}
\newcommand{\flauto}{\ttt{flauto}}
\newcommand{\indicator}{{\mathbb{1}}}
\newcommand{\Dirac}{{\mathsf{\delta}}}
\newcommand{\exponential}{{\mathtt{exp}}}
\newcommand{\normal}{{\mathtt{normal}}}
\newcommand{\uniform}{{\mathtt{unif}}}
\newcommand{\consts}{{\mathsf{Consts}}}
\newcommand{\vars}{{\mathsf{GVars}}}
\newcommand{\procs}{{\mathsf{Procs}}}
\newcommand{\sfproc}{\mathsf{proc}}
\newcommand{\timestamps}{{\nnreals}}

\newcommand{\sys}{{\mathtt{sys}}}
\newcommand{\com}{{\mathtt{com}}}
\newcommand{\execframes}{{\mathsf{Frames}}}
\newcommand{\procdef}{{\mathtt{def}}}
\newcommand{\proc}[3]{{#1(#2)\; =\;\allowbreak #3}}
\newcommand{\postevent}{{\mathsf{post}}}
\newcommand{\postrate}{{\mathsf{post\_rate}}}

\newcommand{\nowcmd}{\mtt{now()}}
\newcommand{\peekcmd}{\mtt{peek()}}
\newcommand{\tick}{\ttt{tick}}

\newcommand{\defeq}{\triangleq}
\newcommand{\post}{\postevent}
\newcommand{\rateSet}{\Lambda}

\newcommand{\tid}{{\mathsf{tid}}}

\newcommand{\dist}{{\mathsf{dist}}}
\newcommand{\unitexp}{{\mathsf{unit}}}

\newcommand{\id}{{\mathsf{id}}}
\newcommand{\newid}{{\mathsf{newid}}}
\newcommand{\now}{{\mathsf{now}}}
\newcommand{\goal}{{\mtt{goal}}}
\newcommand{\Goal}{{\msf{Goal}}}

\newcommand{\execstate}{{\mathsf{exec}}}

\newcommand{\ite}[3]{{\mathsf{if}(#1)\; \mathsf{then}\; #2\; \mathsf{else}\; #3}}

\newcommand{\ceil}[1]{{\left\lceil #1 \right\rceil}}

\newcommand{\DiracNot}[1]{\delta\!\left[ #1 \right]}
\newcommand{\pushforward}{\vartriangleright} 
\newcommand{\bind}{\mathbin{\blacktriangleright}}

\newcommand{\supermartingale}{V}
\newcommand{\spmtg}{\supermartingale}
\newcommand{\variant}{U}
\newcommand{\vnt}{\variant}
\newcommand{\scrS}{\mathscr{S}}
\newcommand{\scrs}{\mathscr{s}}

\newcommand{\fullStateArg}{\frakc} 
\newif\ifrevon \revontrue
\newcommand{\rev}[1]{\ifrevon\textcolor{black}{#1}\else#1\fi}
\newenvironment{revblock}{\ifrevon\color{black}\fi}{}
\newcommand{\revkeep}[1]{\ifrevon\textcolor{black}{#1}\else#1\fi} 

\begin{document}

\title{Formal Reasoning about Performance Models}

\newcommand{\OurInstitution}{Max Planck Institute for Software Systems (MPI-SWS)}
\newcommand{\OurStreet}{Paul-Ehrlich-Stra{\ss}e, Building G26}
\newcommand{\OurCity}{Kaiserslautern}
\newcommand{\OurPostcode}{67663}
\newcommand{\OurCountry}{Germany}


\author{Moussa Labbadi}
\orcid{0000-0002-6109-769X}
\affiliation{%
  \institution{Bretagne INP--ENIB, IRDL (CNRS UMR 6027)}
  \city{Brest}
  \country{France}
}
\email{moussa.labbadi@enib.fr}

\author{Rupak Majumdar}
\orcid{0000-0003-2136-0542}
\affiliation{%
  \institution{\OurInstitution}
  \city{Kaiserslautern}
  \country{Germany}
}
\email{rupak@mpi-sws.org}

\author{V.R. Sathiyanarayana}
\orcid{0009-0006-5187-5415}
\affiliation{%
  \institution{\OurInstitution}
  \city{Kaiserslautern}
  \country{Germany}
}
\email{sramesh@mpi-sws.org}

\author{Sadegh Soudjani}
\orcid{0000-0003-1922-6678}
\affiliation{%
  \institution{University of Birmingham, United Kingdom, and MPI-SWS}
  \city{Kaiserslautern}
  \country{Germany}
}
\email{sadegh@mpi-sws.org}

%
%


\setlist[enumerate,1]{label=(\arabic*)}
\setlist[enumerate,2]{label=(\arabic*)}
\setlist[enumerate,3]{label=(\arabic*)}

\begin{abstract}
Discrete-event simulation  is a standard technique for modelling and analysing the performance of computer systems, networks, and services. 
Although simulation tools are widely used, reasoning about the correctness and performance guarantees of the models they implement 
remains largely ad-hoc: simulation outputs are interpreted statistically, but there is no logical foundation for deductive reasoning about 
their behaviour.
We present a core imperative calculus that captures the essential constructs common to discrete-event 
simulators---asynchronous execution, continuous and discrete sampling from distributions, and time-based event scheduling through a global event queue. 
On top of this calculus, we develop a proof system for reasoning about \emph{almost-sure reachability} and \emph{expected reaching time} properties. 
Our main result is a sound and complete proof rule for these properties.
Our framework generalizes deductive reasoning for (discrete-time) probabilistic programs to the setting of performance models, in
which continuous time and continuous-probability distributions are central.
We have implemented the proof rules in a tool embedded in Lean.
We demonstrate the applicability of our proof rule by deriving proofs of almost sure reachability and expected time for a number of case studies,
including client-server examples that go beyond analytic solutions from queueing theory as well as convergence behaviours in network routing protocols.
Establishing the soundness and completeness of our proof rules requires significantly more complex arguments than in the discrete-time setting.
This is due to the to the fundamentally discontinuous nature of the operational semantics, and to the measure-theoretic challenges of continuous time and probability distributions.
\end{abstract}

\sloppy

\makeatletter
\newtheoremstyle{dazzle}%
{.5\baselineskip\@plus.2\baselineskip
  \@minus.2\baselineskip}
{.5\baselineskip\@plus.2\baselineskip
  \@minus.2\baselineskip}
{\@acmplainbodyfont}
{\@acmplainindent}
{\bfseries}
{.}
{.5em}
{\thmname{\textcolor{red}{\textbf{#1}}}\thmnumber{ \textcolor{red}{\textbf{#2}}}\thmnote{ {\@acmplainnotefont(\textcolor{blue}{#3})}}}
\makeatother

\theoremstyle{dazzle}
\newtheorem{maintheorem}[theorem]{Theorem}
\newtheorem{mainlemma}[theorem]{Lemma}
\newtheorem{maincorollary}[theorem]{Corollary}
\newtheorem{mainproposition}[theorem]{Proposition}

\crefalias{maintheorem}{theorem}
\crefalias{mainlemma}{lemma}
\crefalias{maincorollary}{corollary}
\crefalias{mainproposition}{proposition}

\theoremstyle{acmplain}
\newtheorem{observation}[theorem]{Observation}
\theoremstyle{acmdefinition}
\newtheorem{remark}[theorem]{Remark}

\Crefname{observation}{Observation}{Observations}

\renewcommand{\lstlistingname}{\listingsInLatex}

\maketitle

\raggedbottom
\setlength{\textfloatsep}{10pt plus 2pt minus 2pt}
\setlength{\floatsep}{8pt plus 2pt minus 2pt}
\setlength{\intextsep}{10pt plus 2pt minus 2pt}

\section{Introduction}
\label{sec:intro}

Discrete-event simulation underlies modern performance modeling of computer systems, distributed services, and communication networks
\cite{harchol2013performance,Jain,riley2010ns}. 
Simulation frameworks allow engineers to prototype performance-related algorithms, such as 
scheduling policies, admission control, and load-balancing strategies before service deployment. 
Yet, despite decades of engineering practice, reasoning about these systems' performance remains empirical: 
performance and stability are evaluated by repeated simulation and statistical estimation rather than by formal proofs.
This empirical approach fails to provide guarantees. 
As systems incorporate asynchronous flows and complex real-time feedback, we increasingly need formal reasoning principles
to answer questions related to stability and resilience:
Does the backlog in the system clear almost surely? Is the expected response time finite? 
Does the system remain stable under all arrival rates?
Providing a sound and complete logical framework to answer such questions has been an open challenge.

Classical program logics offer deductive methods for correctness and termination of sequential programs \cite{Hoare69,apt2009verification,MannaComputationBook}. 
However, these frameworks reason about step-based computations without an explicit model of real time or event scheduling. 
In contrast, discrete-event systems feature
(i) a global event queue ordered by timestamps,
(ii) probabilistic and continuous delay distributions, and
(iii) interleaving of asynchronous components through event posting.
Simulation languages define their semantics operationally but provide no proof theory \cite{matloff2008introduction}. 
Existing probabilistic program logics 
handle almost-sure termination 
\cite{McIverMorganBook,McIverMKK18,KaminskiThesis,MajumdarS25}, 
but do not account for event-driven scheduling and continuous-time events. 
Algebraic frameworks \cite{NETKAT,probNETKAT} model finite-state packet-processing behaviors in communication networks, but lack support for continuous sampling and event-driven dynamics.
Consequently, there is currently no sound and complete proof system for reasoning about performance 
and termination properties of probabilistic discrete-event systems.

We address this gap by introducing $\deslang$, a core calculus that abstracts the key mechanisms shared by discrete-event simulators. 
$\deslang$ systems consist of finite sets of procedures; each may perform sequential computation, make discrete probabilistic choices, sample from (continuous) probability distributions, and schedule future procedure calls asynchronously. 
Their execution is driven by a global event scheduler that dequeues the earliest event in time order and executes its procedure body.
Our language captures the essence of simulation engines in a minimal, mathematically precise form suitable for reasoning. 

We study \emph{almost-sure reachability} (is a target reached with probability one?)
and \emph{expected reaching time} for $\deslang$ systems.
These are fundamental questions in performance analysis.
We show \emph{sound and complete} proof rules for these properties.
The rule for almost-sure reachability requires drift and variant functions, similar to discrete-time, countable state models.
The rule for expected reaching time requires a supermartingale function that dominates the hitting time.

While sound and complete proof rules have been studied for discrete-time probabilistic programs, reasoning about $\deslang$ introduces new mathematical
challenges.
When we allow sampling over time and other continuous distributions, the state space is uncountable and can have uncountably many successors.
Thus, our semantics involve the theory of Markov processes over general topological spaces.
Existing techniques related to proving stability of Markov processes assume that the process is continuous (weak Feller) \cite{meyn2012markov, MajumdarSS24}.
Unfortunately, this assumption does not hold for $\deslang$ systems!
The trajectory of a $\deslang$ system for two close states may be very different, e.g., because the code makes a conditional decision or because one  handler is executed before the other.

Dealing with uncountable state and non-continuous dynamics makes our proof very difficult.
Our proof of completeness requires new and subtle topological arguments based on petite sets and their properties \cite{meyn2012markov, xu2018note},
as well as topological arguments (Egorov's theorem) that connects pointwise convergence of measurable functions with uniform convergence.

We have formalized $\deslang$ and our proof rules in $\toolname$.
$\toolname$ embeds $\deslang$ in Lean \cite{Lean}.
Given a $\deslang$ system, together with drift and variant candidates, $\toolname$
provides custom tactics to discharge proof obligations arising out of our proof rules.
We have used $\toolname$ to prove several variations of client–server systems with feedback, timeouts, and retries, 
a canonical performance model used in service-level analysis \cite{harchol2013performance,PEPA,Kleinrock1975}.
Expressed in $\deslang$, this model becomes an infinite-state probabilistic system with scheduled events representing arrivals, completions, and timeouts. 
The almost-sure reachability property we prove consists in proving that some queues in the system empty out almost surely from any configuration;
this corresponds to checking that the system is stable and backlogs are cleared and is a critical performance property in
 request-response systems.
The expected reaching time shows the expected time for backlogs to clear.
Our proof rules are able to show this property, even when the underlying model does not fall into an analytically tractable subclass
of queueing systems.

As a larger case study, we prove a routing protocol from networking \cite{floyd1993synchronization} that shows a surprising 
time synchronization among distributed nodes.
Synchronization occurs even in the presence of jitter---which was added to ensure that nodes do not synchronize!
\cite{floyd1993synchronization} explored this phenomenon using network simulations and simulations of a simple discrete Markov model.
We show the first formal proof that synchronization occurs almost surely in the protocol.
These examples demonstrate that $\deslang$ and $\toolname$ captures many fundamental models in performance evaluation, including
those going beyond analytical closed-form solutions from queueing theory, as well as real protocols for which deductive verification was not possible before.

In summary, we make the following contributions:
\begin{enumerate}
    \item \textbf{Core calculus for discrete-event simulation.} We formalize $\deslang$, a minimal imperative language capturing asynchronous concurrency,
    probabilistic choice, sampling from continuous distributions, continuous-time evolution, and time-based event scheduling.
    We define operational semantics for $\deslang$ that model interaction between probabilistic execution and the global event scheduler,
    and show that these semantics produce Markov processes.
    \item
\textbf{Proof system.} We present the first sound and complete proof rule for reasoning about almost-sure reachability 
and expected reaching time for $\deslang$ systems.

\item 
\textbf{Tool and case studies.} 
We develop $\toolname$, a mechanized proof environment for $\deslang$ embedded in Lean. $\toolname$ formalizes the semantics of DES, mechanizes the soundness of our proof system, and provides tactics for proving almost-sure reachability and expected reaching-time properties. We demonstrate its effectiveness on a number of performance models for
server systems and network routing protocols.
\end{enumerate}
Our results establish a bridge between formal program verification and performance modeling. 
By providing a sound and complete logical foundation for discrete-event systems, we move 
beyond simulation-based validation toward provable performance guarantees and open the way for formal reasoning tools for system performance.

\section{Reasoning about Discrete Event Simulation Models}
\label{sec:motivating-examples}


We will work with $\deslang$, an imperative language that captures the core features of discrete-event simulators
such as SimPy \cite{matloff2008introduction}.
This language describes systems as sets of procedures, each implementing an imperative core.
In addition to usual imperative constructs, procedures can perform discrete probabilistic choice, sample from distributions, and schedule procedures for invocation at specific times in the future.
Execution is driven by an event scheduler that maintains a global event queue ordered by time.
Procedures execute sequential probabilistic computation atomically, and in logical zero time.
They post new procedures to be executed by the scheduler at future points in time.
The event scheduler tracks global time.
When a procedure finishes execution, it picks the minimal next time at which a procedure is scheduled, 
moves time forward to that point, and runs that procedure. 
By always picking the minimal event time, the discrete-event simulation naturally proceeds in correct temporal order.

\begin{system}[t]
\begin{footnotesize}

\begin{minipage}{0.45\textwidth}
\begin{lstlisting}[mathescape=true]
proc Client() =
  req := ... ; // generate request
  $\textcolor{blue}{\mathbf{\mathtt{post}}}$(SendReq, req, $\textcolor{brown}{\Dirac}(0)$);
  // schedule next arrival
  $\textcolor{blue}{\mathbf{\mathtt{post}}}$(Client, $\textcolor{brown}{\exponential}(\lambda)$);

proc SendReq(req) =
  // enqueue request
  queue := enqueue(queue, req);
  // post Server task
  $\textcolor{blue}{\mathbf{\mathtt{post}}}$(Server, $\textcolor{brown}{\Dirac}(0)$);

proc $p_0$():
  queue, idle := [], true;
  $\textcolor{blue}{\mathbf{\mathtt{post}}}$(Client, $\textcolor{brown}{\exponential}(\lambda)$);
\end{lstlisting}
\end{minipage}
\begin{minipage}{0.45\textwidth}
\begin{lstlisting}[mathescape=true,firstnumber=16]
proc Server() =
  if idle and not queue_empty then
    idle := false;
    req := dequeue(queue);
    $\textcolor{blue}{\mathbf{\mathtt{post}}}$(FinishReq, req, $\textcolor{brown}{\exponential}(\mu)$);

proc FinishReq(req) =
  ... ; // service completion
  if not queue_empty then
    // schedule more tasks
    req := dequeue(queue);
    $\textcolor{blue}{\mathbf{\mathtt{post}}}$(FinishReq, req, $\textcolor{brown}{\exponential}(\mu)$);
  else
    idle := true;
\end{lstlisting}
\end{minipage}
\end{footnotesize}
\caption{A client server example: clients send requests based on an exponential distribution with rate $\lambda$,
server completion time is an (independent) exponential distribution with rate $\mu$, requests are queued in an unbounded
queue. The system is an encoding of an M/M/1 queue in $\deslang$. \label{fig:mm1}}
\end{system}
\noindent\textbf{A Client-Server Model}
\noindent\codeRef{fig:mm1} shows a simple simulation model for a client-server example.
The \ttt{Client} procedure models request arrivals: it schedules a request to the server right away and also schedules the next request arrival after an exponentially distributed delay $\exponential(\lambda)$, creating an infinite Poisson arrival process.
The \ttt{SendReq} procedure enqueues the request at the server and schedules the server to serve it.
When scheduled, the server dequeues a request and serves it if it is idle (i.e., not already serving a different request).
Serving a request takes time; in the model, the time is drawn from an independent exponential distribution $\exponential(\mu)$.
Upon completion, if there are further tasks in the queue, the server picks a new task.
Otherwise, the state variable \ttt{idle} is set, which signifies that the server is not currently processing a request.

This example is the ``hello world'' of performance simulation.
It already shows several features that we model:
(a) imperative sequential computation, (b) asynchronous invocation, (c) real-time evolution, (d) (continuous and discrete) probabilistic choice.
Since we have not put a bound on the queue size, the program is also infinite-state.
The combination of these features takes the example outside of known program logics.

\smallskip
\noindent
\textbf{Almost-Sure Reachability: Specifications and a Proof Rule}
An important question in a client-server example is whether the server has enough capacity to serve arriving requests.
One way to formulate this question is to ask whether the queue becomes empty almost surely, from any starting configuration.
We are also interested in showing the expected time before this event happens, which answers the question how long the system
takes, on average, to clear backlog.
%

Our proof rule will ask for two kinds of certificates.
First, 
a \emph{drift certificate} that will assign a non-negative number to each state of the system that will not increase in
expectation on each step of the program.
Second, a \emph{variant certificate} will assign a non-negative number to each state, and require two conditions:
(a) the variant decreases in each step with some positive probability bounded away from zero,
and (b) the variant at any state is bounded by some (known) function of the drift at that state.
Our proof rule will show that if these two conditions hold, then the program will reach a given state almost surely.

Our goal in this example is to prove that the queue becomes empty almost surely.
Accordingly, we use the size of the \ttt{queue} as candidate to both the drift and the variant.
Under the assumption that $\lambda < \mu$, we will show that this queue size satisfies the drift and variant conditions.
On a global level, the execution of our example is a sequence of events that either add or remove elements from the queue.
The queue grows when \ttt{Client} fires and shrinks when \ttt{FinishReq} fires.
Because $\lambda < \mu$, the queue shrinks at a greater rate than it grows, and hence, the size of the queue drifts toward zero, making it a valid drift.
Similarly, because $\lambda < \mu$, at each event, it is more likely for the queue to shrink; thus, the queue size is also a valid variant.
Note that this size is clearly non-negative, and zero only at goal states.
Thus, the queue empties almost surely.
Under the same stability assumption $\lambda < \mu$, our expected-time proof rule additionally constructs a 
supermartingale that upper-bounds the expected time until the queue first becomes empty, showing that the expected backlog-clearing time is finite.
We provide exact details in \cref{sec:specs}.

While classical queueing theory can derive these particular bounds, the strength of our approach is that we can prove similar properties even in cases where closed form solutions are not available: for example, when we add common features such as throttling, retries, or exponential backoff.

\section{The $\deslang$ Language}
\label{sec:prelims}

We now formally define $\deslang$, a core language for simulation, and its semantics as a Markov process.

\subsection{Syntax}
\Cref{fig:syntax} shows the syntax of our language $\deslang$.
Systems in $\deslang$ are defined by a set of procedures that includes a distinguished initial procedure.
Each procedure is a command that is paired with a set of variable parameters.
Variables take on values from the set of real numbers $\reals$.
Commands are imperative blocks of code that randomly manipulate these real-valued variables.
$\deslang$ offers the standard imperative control constructs and probabilistic choice operators inspired by existing probabilistic programming languages \cite{McIverMorganBook,Kozen85}.
These are $\oplus_p$ and $\sim$, and they enable discrete and continuous probabilistic choice respectively.
Note that $\deslang$ does not offer support for the looping constructs such as \textsf{while}; nonetheless, loops can be simulated using the event queue detailed below.


\begin{figure}[t]
\begin{align*}
\mbox{\textsc{Systems} }\quad & \sys\ ::=\ \set{\ \procdef_1, \procdef_2, \ldots, \procdef_k,\; \mathsf{proc }\ \proc{p_0}{}{\com}\ } \\  
\mbox{\textsc{Procedures} }\quad & \procdef\ ::=\ \mathsf{proc }\ \proc{p}{y_1, y_2, \ldots, y_k}{\com} \\
\mbox{\textsc{Commands} }\quad & \com\ ::=\  \sfskip \mid x :=\ e \mid x \sim D \mid \com_1;\, \com_2 \mid\com_1 \oplus_p \com_2 \\ 
                    & \quad\quad        \mid \ite{b}{\com_1}{\com_2} \mid 
                    \post(p, e_1, e_2, \ldots, e_k, D_t) \\
                    & \quad\quad        \mid x := \nowcmd \mid x := \peekcmd
\end{align*}
\caption{\emph{Syntax of $\deslang$.}
We assume syntactic categories for constants $c\in\consts$, 
global variables $x\in\vars$, local variables $y_i \in \lVars$, 
procedure names $p\in\procs$, 
and probability distributions $D$ and $D_t$ over $\reals$ and $\nnreals$ respectively.
$p_0$ is a distinguished member of $\procs$ used for the initial procedure.
$e_i$ and $b$ are arithmetic and boolean expressions over $\Vars \cup \consts$.
\label{fig:syntax}}
\end{figure}



$\deslang$ provides primitives to insert procedures into a global event queue to be executed at a later time.
Time is real-valued, and is measured by the non-negative reals $\nnreals$.
The operation of $\deslang$ dequeues the earliest procedure from this event queue for execution.
The primitive $\post$ schedules procedures by sampling a real-time delay from a distribution over $\nnreals$.
The command $\post(p, e_1, e_2, \ldots, e_k, D_t)$ schedules the execution of the procedure $p$ with arguments $e_1, e_2, \ldots, e_k$, 
at the future time $\now + \tau$,  where $\now$ is the (global) current time and $\tau \sim D_t$ 
is sampled from the distribution $D$ over $\nnreals$.
We assume the existence of independent syntactic categories for describing distributions over $\nnreals$ that include standard distributions such as the exponential distributions $\exponential$,
Dirac (point) distributions $\Dirac$, uniform distributions $\uniform$, and Gaussian distributions $\normal$. 
In addition to $\post$, $\deslang$ provides primitives $\nowcmd$ and $\peekcmd$ to allow commands read-only access to the event queue.
These reveal the current global time $\now$ and the scheduled time of the next pending event in the queue.
Using this event queue, a $\deslang$ system can simulate loops by posting continuation with delay $0$ using $\Dirac(0)$.


\subsection{Markov Processes as Operational Semantics}

We will define the operational semantics of $\deslang$ as discrete-time \emph{Markov processes}.
In the following, we assume some familiarity with measure theory and probability,
such as sigma algebras, probability spaces, filtration, and expectation of random variables (see, e.g., \cite{Billingsley} for background).

\begin{definition}
    A discrete-time \emph{Markov process (MP)} is a tuple $\mathfrak{S} = \left(X, T_X\right)$, where $X$ is any set referred to as the \emph{state space}, is equipped with a sigma algebra $\calB(X)$, and $T_X \colon X \times \mathcal{B}(X) \to [0,1]$ is a conditional stochastic kernel 
    that assigns to each state $x \in X$ a probability measure $T_X(x, \cdot)$ over the measurable space $(X, \calB(X))$. 
    Additionally, for every measurable set $A \in \calB(X)$, the function $x \mapsto T_X(x, A)$ is measurable.
\end{definition}

In addition to $T_X$, we use $T^n_X$ to denote the $n$-step transition kernel, where $T^1_X = T^n_X$ and $T^{n + 1}_X(x, A) = \int_X T_X(x, dy) T^n_X(y, A)$ for all $n \in \setOfNaturals$.

The dynamics of a Markov process unfold as follows.
Let $x_k$ denote the state at time step $k \in \mathbb{N}$.
Given the current state $x_k = x \in X$, the next state $x_{k+1}$ is drawn according to the probability measure $T_X(x, \cdot)$.
For each process, we fix a distinguished initial state $x_0 \in X$.
We write $\mathbb{P}$ to denote the probability measure over infinite trajectories $\xi := (x_0, x_1, x_2, \ldots)$ induced by $T_X$.



The semantics of $\deslang$ systems define discrete-time Markov processes.
We now give the notation we need to define these semantics.
A \emph{program state} is a partial function that maps a finite subset of $\Vars \cup \LVars$ to $\reals$.
We write $\Sigma$ for the set of all program states, 
$\calS$ for the set of all $\deslang$ systems derivable using the syntax in \cref{fig:syntax}, and
$\calC$ for the set of all commands derivable in $\deslang$.

A \emph{pending event} is a tuple 
$(t, \id, p, \bfv)$, where 
$t \in \nnreals$ is the scheduled time of the event, 
$\id$ is an unique (linearly-ordered) event identifier used for deterministic tie-breaking, 
$p \in \procs$ is a procedure name, and 
$\bfv = (v_1, \ldots, v_k)$ are the values of the arguments of $p$.
We assume that the event identifiers $\id$ come from $\nats$.
Intuitively, a pending event corresponds to a call to procedure $p$ with arguments $\bfv$ scheduled to be executed at time $t$.

The operational semantics of $\deslang$ maintains finite priority queues $Q$ of pending events.
Ordering on $Q$ is lexicographic on $(t, \id)$, i.e., for two equal timestamps, the lower $\id$ wins.
We write $\calQ$ for the set of all finite priority queues of events 
$(t, \id, p, \bfv) \in \timestamps \times \nats \times \procs \times \cup_{n \in \nats} \nnreals^n$.

A $\deslang$ \emph{configuration} is a tuple $\tuple{Q, S, \sigma, \now, \sfc}$, where 
$Q$ is a priority queue of pending events such that $t \geq \now$ for all scheduled times of pending events in $Q$, 
$S \in \calS$ is a global system environment containing the syntax of a finite set of procedures, 
$\sigma \in \Sigma$ is a program state that maps all global and local variables used in $S$, 
$\now \in \timestamps$ is the current time, and 
$\sfc\in \calC$ is a $\deslang$ command representing the state of the execution of the current procedure.
Configurations are denoted by $\frakc$, and we use $\frakC$ to denote the set of all configurations.
The initial configuration of a $\deslang$ system $S$ is $\frakc_0 = (\emptyset, S, \sigma_0, 0, \com_{p_0})$, where $\com_{p_0}$ is the command in $S$ ascribed to the initial procedure $p_0$ and $\sigma_0$ is a distinguished initial program state.
For a configuration $\frakc$, let $\frakc.Q$, $\frakc.S$, $\frakc.\sigma$, $\frakc.t$, and $\frakc.\sfc$ denote its queue, $\deslang$ system, program state, time, and command, respectively.

%

\paragraph{Informal Overview: One Step}
The operational semantics of $\deslang$ is defined by a relation $\frakc \Rightarrow \mu$ that assigns to configuration $\frakC$ a distribution $\mu$ over configurations that arise in one step.

Let us first intuitively describe these steps.
A one-step in a $\deslang$ system can be classified into two kinds:
(1) a \emph{scheduler dequeue step} that pops the event with the minimal scheduled time and id (i.e., with the highest priority) from the event queue for execution, and
(2) an \emph{execution step} of the body of a procedure that can insert new events into the queue and transforms program state.

A scheduler dequeue step is performed when the queue $Q$ is not empty and the current command $\sfc$ is $\sfskip$.
Under these conditions, the scheduler picks the event $\min(Q) = (t_{\min}, \id_{\min}, p, \bfv)$ in $Q$ with minimal $(t, \id)$ in the lexicographic ordering of $Q$.
This transitions the execution to a new configuration $(Q \setminus \{\min(Q)\}, S, \sigma', t_{\min}, \com_p)$ with $\com_p$ being the body of the procedure $p$ defined in $S$ and $\sigma'$ produced by augmenting $\sigma$ with the values of $v_i$ for each parameter in the definition of $p$.
Thus, the time jumps ahead to $t_{\min}$ and the body $\com$ of the procedure $p$ is executed with arguments $\bfv$.
If $p$ is not described in $S$, $\com = \sfskip$.


The body of every procedure is a loop-free imperative probabilistic program with side-effects
enabled by $\post$ that insert elements into $Q$.
Note that, because each command is loop-free, every procedure always terminates.
Thus, without $\post$, their semantics is a natural generalization of standard operational semantics for straight-line probabilistic programs. 
These semantics are standard for the deterministic constructs \cite{McIverMorganBook,KaminskiThesis}. 
The probabilistic command $\com_1 \oplus_p \com_2$ executes $\com_1$ with probability $1-p$ and $\com_2$ with probability $p$.
The sampling command $x \sim D$ samples a value $v$ from $\reals$ according to the probability distribution $D$ and assigns $v$ to the variable $x$.
This behaviour is captured by the operational semantics shown in \cref{fig:semantics-des}.


Posts are the novel feature of $\deslang$.
They perform probabilistic sampling from continuous distributions over time and produce a queue side-effect.
Suppose $\com = \post(p', \bfe', D)$ in the current configuration $\tuple{Q, S, \sigma, \now, \com}$, where $\bfe' = (e'_1, \ldots, e'_k)$.
Suppose the evaluation of the expressions $e_i$ yield values $v_i$; these are the aforementioned deterministic evaluations. 
Then the post step is a probabilistic transition that yields a probability measure $\mu$ over configurations that maps any measurable set of configurations $A$ to
\[
\mu(A) = \int_\timestamps \indicator_A\set{\tuple{Q \uplus \set{(\now + \tau, \newid', p', \bfv')}, S, \sigma', \now, \sfskip}}  dD(\tau)
\]
where $\newid' \in \nats$ is a fresh id greater than all identifiers in $Q$ and $\bfv' = (v_1, \ldots, v_k)$.
Intuitively, we sample a delay $\tau$ from the distribution $D$, create a fresh monotone identifier $\newid'$ (to preserve deterministic tie-breaking), and return the point configuration with the queue extended by $(\now + \tau, \newid', p', \bfv')$  and the post command evaluated to $\sfskip$.
The effect is a probabilistic enqueueing of exactly one event. 

\begin{figure}[t]
\footnotesize
\begin{mathpar}
\inferrule[assign]{ }
{
        (Q, S, \sigma, t, x := e) \Rightarrow \DiracNot{(Q, S, \sigma[x \mapsto \sigma(e)],t, \sfskip)}
}
\and
\inferrule[sample]{ }
{
        (Q, S, \sigma, t, x \sim D) \Rightarrow \lambda v.\, (Q, S, \sigma[x \mapsto v], t, \sfskip) \pushforward \sigma(D)
}
\\
\inferrule[now]{ }
{
        (Q, S, \sigma, t, x := \nowcmd) \Rightarrow \DiracNot{(Q, S, \sigma[x \mapsto t], t, \sfskip)}
}
\and
\inferrule[prob]{
        (Q, S, \sigma, t, \sfc_1) \Rightarrow \mu_1 \\ (Q, S, \sigma, t, \sfc_2) \Rightarrow \mu_2
}{
        (Q, S, \sigma, t, \sfc_1 \pChoice{p} \sfc_2)  \Rightarrow  (1 - \sigma(p)) \mu_1 + \sigma(p) \mu_2
}
\\
\inferrule[peek]{ }
{
        (Q, S, \sigma, t, x := \peekcmd) \Rightarrow \DiracNot{(Q, S, \sigma[x \mapsto \mathsf{peek}(Q)], t, \sfskip)}
}
\\
\inferrule[if\textsubscript{1}]{
        \sigma(b) = \top \\ (Q, S, \sigma, t, \sfc_1) \Rightarrow \mu_1
}{
        (Q, S, \sigma, t, \ite{b}{\sfc_1}{ \sfc_2}) \Rightarrow \mu_1
}
\and
\inferrule[if\textsubscript{2}]{
        \sigma(b) = \bot \\ (Q, S, \sigma, t, \sfc_2) \Rightarrow \mu_2
}{
        (Q, S, \sigma, t, \ite{b}{\sfc_1}{ \sfc_2}) \Rightarrow \mu_2
}
\\
\inferrule[concat]{
        (Q, S, \sigma, t, \sfc_1) \Rightarrow \mu \\ \forall Q'\, \forall \sigma' \cdot (Q', S, \sigma', t, \sfc_2) \Rightarrow \cK(Q', \sigma')
}{
        (Q, S, \sigma, t, \sfc_1;\ \sfc_2) \Rightarrow
        \mu \bind \lambda \frakc.\, \cK(\frakc.Q, \frakc.\sigma) 
}
\and
\inferrule[deque]{
        Q \neq \emptyset \\ \min(Q) = (t', \id, p, \bfv) \\ t' \geq t \\ S(p) = \sfc
}{
        (Q, S, \sigma, t, \sfskip) \Rightarrow \DiracNot{(Q \setminus \min(Q), S, \sigma[\bfy \mapsto \bfv], t', \sfc)}
}
\\
\inferrule[post]{
        \forall (t', \id', p', \bfv') \in Q \cdot \id > \id'
}{
        (Q, S, \sigma, t, \post(p, \bfe, D)) \Rightarrow  \lambda \tau.\, (Q \uplus \set{(t + \tau, \id, p, \sigma(\bfe))}, S, \sigma, \sfskip) \pushforward \sigma(D)
}
\end{mathpar}
\caption{
    \emph{Operational Semantics of $\deslang$.}
    The relation $\Rightarrow$ assigns to each configuration a probability distribution over the set of all configurations.
    $\sigma[x \mapsto \sigma(e)]$ is obtained by setting the variable $x$ in $\sigma$ to the value of the expression $e$ in $\sigma$; $\sigma[\bfy \mapsto \bfv]$ generalizes this to a tuple of variables $\bfy$.
    $\DiracNot{\frakc}$ is the Dirac (point) distribution that assigns probability $1$ to the singleton $\set{\frakc}$. 
    $\mathsf{peek}(Q)$ is the timestamp of $\min(Q)$, or $\infty$ (a designated top element of $\timestamps$) when $Q$ is empty.
    $f \pushforward \mu$ is a pushforward operation that samples the output of the function $f$ according to the distribution $\mu$.
    $\mu \bind \cK$ is the 
    distribution produced by sampling a value according to $\mu$ and taking the distribution associated to it by $\cK$.
}
\label{fig:semantics-des}
\end{figure}

%

\begin{example}
Consider a system $S$ consisting of two procedures 
$\sfproc\ \allowbreak \proc{p_0}{}{\allowbreak\post(B,  1, \dist(\exponential(\lambda))}$ and 
$\sfproc\ \proc{B}{y}{\sfskip}$. 
The execution of this system proceeds as follows.
The initial configuration of the system is $({0, \id_0, p_0}, S, \sigma_0, 0, \sfskip)$.
The scheduler  dequeues $(0, \id_0, p_0)$, sets $\now = 0$, and then begins executing the initial procedure $p_0()$.
This moves the execution to $(\set{}, S, \sigma_0, 0, \post(B,  1, \dist(\exponential(\lambda)))$.
The semantics samples a delay $\tau \sim \exponential(\lambda)$ and enqueues $(0 + \tau, \id_1, B, 1)$, where $\id_1$ is a fresh symbol greater than $\id_0$.
Then, the scheduler picks the minimal time $\tau$, sets $\now = \tau$ and starts executing $B(1)$.
Since this does not post any new tasks, the program terminates after $B$ finishes executing.
Note that time only increases in a scheduler step.
If several $\post$ calls produce different sampled delays, the scheduler always selects 
the event with the smallest sampled scheduled time next.
\hfill\qed
\end{example}


\paragraph{One-Step Operational Semantics: Formal Details}
\label{subsec:semantics-des}

To formally define the semantics, we need two additional notions.
The first is the \emph{pushforward measure}.
Given two measurable spaces $(X_1, \calB(X_1))$ and $(X_2, \calB(X_2))$, a measurable function $f : X_1 \to X_2$, and a measure $\mu$ over $(X_1, \calB(X_1))$, the pushforward measure $f \pushforward \mu$ is the measure given by
$$
(f \pushforward \mu)(A) \triangleq \mu(f^{-1}(A)), \quad\forall A \in \calB(X_2).
$$
Intuitively, $f \pushforward \mu$ is the distribution of $f(x)$ over $X_2$ when $x$ is sampled from $X_1$ according to $\mu$.

The second is the composition of a kernel and a measure. 
In a measurable space $(X, \calB(X))$, a kernel is a function $\cK$ that assigns a probability measure over $(X, \calB(X))$ to each element in $X$.
Thus, the composition of a kernel $\cK$ and a measure $\mu$ over the space is given by 
$$
(\mu \bind \cK)(A) = \int_X \cK(x)(A) d\mu(x).
$$
Intuitively, the measure $\mu \bind \cK$ is obtained by first sampling a point $x\in X$ according to $\mu$ and then taking the resulting distribution over $X$ according to $\cK(x)$.

Let $\calB(\frakC)$ represent the Borel $\sigma$-algebra over $\frakC$.
The operational semantics of $\deslang$ is given by the relation $\Rightarrow$ that recursively assigns a probability measure over the space $(\frakC, \calB(\frakC))$ to each configuration in $\frakC$, as shown in \cref{fig:semantics-des}.
Thus, $\frakc \Rightarrow \mu$ means that the semantics assigns the measure $\mu$ over $(\frakC, \calB(\frakC))$ to the configuration $\frakc$.

Notice that configurations with $\sfskip$ commands and empty queues are not related by $\Rightarrow$; thus, these configurations are \emph{terminal} configurations.
Moreover, an assignment of a full probability measure as the effect of each command is only possible because every command that can be written in $\deslang$ is loop-free and thus, always terminating.
For every configuration $\frakc$ with a non-$\sfskip$ command, the distribution $\mu$ assigned by $\frakc \Rightarrow \mu$ assigns measure $1$ to the set of configurations with a $\sfskip$ command.
Thus, the small-steps of commands are compiled into a single $\deslang$ step.


\paragraph{From Steps to Traces: Markov Processes Semantics}
\label{subsec:markov-semantics}
The one-step relation $\Rightarrow$
induces a discrete-time Markov process $(\frakC, T_\frakC)$ evolving in the measurable space of all configurations $(\frakC, \calB(\frakC))$.
%
For any configuration $\frakc$ with $\frakc \Rightarrow \mu_\frakc$ according to the semantics defined in \cref{fig:semantics-des}, the transition kernel is defined as $T_\frakC(\frakc, A) = \mu_\frakc(A)$ for all $A \in \calB(\frakC)$.
For every terminal configuration $\frakc$, define $T_\frakC(\frakc, \set{\frakc}) = 1$.
It is easy to see that for every $A \in \calB(\frakC)$, the function $\frakc \mapsto T_\frakC(\frakc, A)$ is measurable.
We denote by $\pp_S$ probability measure over the infinite trajectories of this chain induced by $T_\frakC$.

\section{Almost-Sure Reachability and Expected Reach Time for $\deslang$}
\label{sec:specs}


We are interested in proving the almost-sure reachability of distinguished goal states in Markov processes and $\deslang$ systems.
The almost-sure reachability question asks: does a $\deslang$ system $S$ reach a measurable set $\Goal$ with probability $1$?

Formally, let $(X, T_X)$ be a discrete-time Markov process.
Let $\Goal$ be a distinguished measurable subset of $X$.
Define the \emph{first hitting time} $\sigma_{\Goal}$ as 
\begin{equation}
    \sigma_{\Goal} = \begin{cases}
        \min\{k: x_k \in \Goal\} & \text{if } \exists k \in \nats \text{ with }x_k\in \Goal,\\
        \infty & \text{if } \forall k\in\nats, x_k\notin \Goal.
    \end{cases}
\end{equation}
Here, $x_k$ is a random variable equal to the state reached in $k$ steps 
and $x_0$ is the initial state.
The \emph{reachability probability} of $\Goal$ is the probability $\pp(\sigma_{\Goal} < \infty)$ that $\sigma_{\Goal}$ is finite.
We say $(X, T_X)$ \emph{almost surely reaches $\Goal$} if $\pp_S(\sigma_{\Goal} < \infty) = 1$. 

We can state the same problems on $\deslang$ systems.
Let $S$ be a $\deslang$ system with a distinguished variable $\goal$ initialized to $0$.
The \emph{almost-sure reachability question} asks: does $S$ reach a configuration in which $\goal = 1$ with probability $1$?
Let $\Goal$ be the set of configurations $\frakc$ with $\frakc.\sigma(\goal) = 1$. 
Then, $S$ reaches $\goal$ almost surely \emph{iff} the Markov process $(\frakC, T_{\frakC})$ reaches the set $\Goal$ almost surely. 
The semantics of $\deslang$ as Markov processes implies that these notions are measure-theoretically well-defined.

In addition to the almost sure reachability question, we are interested in the \emph{expected time} to reach the goal.
In $\deslang$ systems, time is a part of the configuration---meaning the $\deslang$ time is subsumed into the states of the Markov process defining its semantics.
Thus, expected time of $\deslang$ systems is not the expected value $\bbE[\sigma_{\Goal}]$ of the \emph{number of steps} in the Markov process; it is the expected value of the time in configurations reached with $\frakc.\sigma(\goal) = 1$.
Formally, for the Markov process $(\frakC, T_\frakC)$ defining a $\deslang$ system, let $T_k$ be a random variable equal to the value of the $\deslang$ time reached in $k$ steps.
Then, the expected time of a $\deslang$ system is given by $\bbE[T_{\sigma_{\Goal}}]$.

\subsection{Proof Rules for Almost-Sure Reachability: Overview}

In this section, we will define proof rules to reason about almost-sure reachability and upper bounds on the expected reach time.
Proof rules are deduction rules that form a part of a proof system within a formal program logic \cite{winskel}.
They are useful in enabling the deduction of complex program properties, including termination \cite{MajumdarS25,McIverMKK18} and reachability \cite{MajumdarSS24}.
They are composed of finite collections of \emph{certificates} (a.k.a.\ witnesses): functions that map program states to numbers that satisfy certain properties.

We will introduce our proof rules in two steps.
We will begin with proof rules for almost sure reachability for general Markov processes, and show that they are sound and complete.
This is a general theorem of independent interest in the theory of Markov processes.
The technical difficulty over related results is that general Markov processes underlying $\deslang$ systems do not satisfy continuity assumptions imposed by available results \cite{MajumdarSS24,meyn2012markov}.

We then describe a compositional proof rule for $\deslang$ systems that provides certificates per procedure.
We argue about the soundness and completeness of this rule by translating the certificates 
to and from the Markov processes defining their semantics.

For both Markov processes and for $\deslang$ systems, our proof rule will be made up of several classes of witnesses:

\smallskip
\noindent\textit{The Invariant.}
In program verification, an invariant is a Boolean function that, if true at a state, will remain true at every successor of that state.
For $\deslang$ systems, an invariant $\Inv$ is a measurable set of configurations such that $\frakc_0 \in \Inv$ and, for all $\frakc \in \Inv$, $\mu(\frakc, \Inv) = 1$ where the semantics induces $\frakc \Rightarrow \mu$.

\smallskip\noindent\textit{Drift or Supermartingales.}
A supermartingale is a non-negative function $\spmtg$ of the full state that, in expectation, does not increase at each execution step.
In probability theory, supermartingales are stochastic processes whose value is non-increasing in expectation throughout its progression.
Supermartingales have a long history in proving properties of probabilistic programs \cite{abate2025quantitative,ChakarovSankaranarayanan}, and have recently been used to characterize their almost-sure termination \cite{McIverMKK18,MajumdarS25}.

\smallskip\noindent\textit{Variants}
A (probabilistic) variant is a non-negative function $\vnt$ defined on full states that decreases at each execution step by a minimum value with a minimum probability.
These variants are the probabilistic generalization of ranking functions, which for non-probabilistic programs, are classical certificates for proving program termination \cite{AptP86,winskel,apt2009verification}.
Variants are essential in the effectiveness of supermartingales, and work together to demonstrate termination of probabilistic programs \cite{McIverMKK18,MajumdarS25}.

These are the familiar constructs from the theory of countable Markov chains \cite{McIverMKK18,MajumdarS25}.
This is to be expected, because our results subsume those.

\begin{figure}[t]

\begin{proofrule}[Almost-Sure Reachability for Markov processes]
\label{rule:des-ast-markov}
To show that the chain $(X, T_X)$ almost surely reaches a measurable target $\Goal \subseteq X$, find 
\begin{enumerate}[label=\arabic*.,leftmargin=1.5em]
    \item A measurable \emph{invariant} $\Inv \supseteq \Goal$ containing the initial state $x_0$,
    \item A \rev{measurable} \emph{exempt set} $T$ with $\Goal \subseteq T \subseteq \Inv$,
    \item A \rev{measurable} \emph{supermartingale $\spmtg$ and variant} $\vnt$, both mapping $X \to \nnreals$, and
    \item Assistant functions $H, d, \epsilon : \nnreals \to \preals$, 
\end{enumerate}
such that the following are true:
\paragraph{$\mbf{V1:}$ Drift Criterion.} The supermartingale $\spmtg$ vanishes exactly on the exempt set, $V(x) = 0 \Leftrightarrow x \in T$, and does not increase outside it:
    \begin{equation*}\label{eq:diff}
    \Delta V(x) := \int_\Inv T_X(x, dy)\,V(y) - V(x) \leq 0, \quad \forall x \in \Inv \setminus T.
    \end{equation*}

\paragraph{$\mbf{V2:}$ Variant Criterion.} For all \( r \in \nnreals \) and \( x \in \Inv \),
\[
V(x) \leq r \Rightarrow U(x) \leq H(r), \quad\quad U(x) = 0 \Leftrightarrow x \in \Goal, \qquad \text{and}
\]
\[
T_X(x, \set{y \in \Inv \mid U(y) \leq U(x) - d(r)}) \geq \varepsilon(r), \quad\forall x \in \Inv\setminus\Goal \;\text{with}\; \rev{U(x) \leq r}.
\]
\end{proofrule}

\caption{The rule for deducing almost-sure reachability for Markov processes.}
\label{fig:rule-markov}
\end{figure}

\subsection{Proof Rules for Almost-Sure Reachability: Markov Processes}
\label{subsec:rule-ast-markov}
\cref{rule:des-ast-markov} shows our proof rule.
The rule asks for a \rev{measurable} supermartingale $\spmtg$ and a \rev{measurable} variant $\vnt$ over the state space.
The invariant $\Inv$ asked by the rule is a measurable subset of $X$ such that $x_0\in \Inv$ and $T_X(x, \Inv) = 1$ for all $x\in \Inv \rev{{}\setminus \Goal}$.
\rev{From a goal state, the chain may leave the invariant with positive probability.} 
The \emph{exempt set} $T$ is a measurable set between $\Goal$ and $\Inv$; it collects the states near $\Goal$ where the supermartingale $\spmtg$ is $0$.

The condition \textbf{(V1)} ensures that $V$ is a supermartingale that is zero on the exempt set (and hence on the goal), 
and that it is non-increasing in expectation outside the exempt set.
The condition \textbf{(V2)} ensures that the variant is bounded on every sublevel set of the supermartingale,
that it is zero on the goal set, and that it decreases with positive probability bounded away from zero \rev{within each of its own sublevel sets}.
%
Observe that $\spmtg \equiv 0$ on $T$, so instantiating \textbf{V2} at $r = 0$ leaves $\vnt \leq H(0)$ throughout $T$.


\rev{Note that the rule of \citet{MajumdarSS24} differs from \cref{rule:des-ast-markov} by indexing its functions $d$ and $\epsilon$ with the supermartingale instead of the variant.
This choice of indexing makes their rule unsound, as we show in our supplementary material.}

A proof rule for establishing almost-sure reachability is \emph{sound} if, whenever its premises hold for a Markov process $(X, T_X)$, its trajectory reaches $\Goal$ almost surely.
It is \emph{complete} if, for every Markov process $(X, T_X)$ that reaches $\Goal$ almost surely,  
witnesses matching the premises of the rule exist.

\begin{theorem}
\label{th:ast-markov}
\cref{rule:des-ast-markov} is sound and complete.
\end{theorem}

\begin{proof}[Proof Discussion]
The proof is very technical.
We only give intuitions here; see the supplementary material for details.
The proof of soundness is an extension of the techniques from \cite{McIverMKK18,MajumdarS25,MajumdarSS24}.
Their completeness proofs however require assumptions not met by $\deslang$; thus, we use new techniques to prove completeness.

At a high level, our rule is sound for the following reasons.
First, the supermartingale \( \spmtg \) in Drift Condition \textbf{V1} provides a means to partition the space of configurations \( X \) 
into concentric ``circles'' around \( \Goal \).
For example, the sublevel set \( \{x \in X \mid \spmtg(x) \leq 0\} \) is exactly the exempt set $T$, the ``innermost circle'', \( \{x \in X \mid \spmtg(x) \leq 1\} \) defines the next circle, and so on.
Because $\spmtg$ is required to be a supermartingale only \emph{outside} $T$ and vanishes \emph{on} $T$, the stopped process $\spmtg(x_{k \wedge \tau_T})$, where $\tau_T$ is the first time the system enters $T$, is a non-negative supermartingale.
The Martingale Convergence Theorem of \citet{doobBook} then ensures the system does not escape to infinitely many higher circles. 

The variant function \( \vnt \) is used to measure progress within these circles.
The Variant Condition \textbf{V2} implies that $\vnt$ will be bounded from above in each sublevel set of $\spmtg$; 
in particular, since $\spmtg \equiv 0$ on $T$, the variant is bounded on $T$. 
Such bounded variants are effective for showing reachability within these sublevel sets \cite[Lemma 7.5.1]{McIverMorganBook}.
\rev{This means that, whenever the system is inside a given circle (i.e., a bounded sublevel set of $V$), it reaches $\Goal$ with a non-zero minimum probability at each circle.
Because the system almost surely revisits some circle infinitely often, a zero-one law of probabilistic processes gives the soundness of the rule.}

The completeness proof is much more challenging, and we briefly review why existing completeness proofs do not imply our results.

Using techniques from Markov chain theory \cite{Foster51,Foster53,MSZ78}, completeness for almost-sure termination was only established recently \cite{MajumdarS25} for probabilistic programs evolving over countable state spaces.
Their proof critically relies on the countability of the state space, allowing an enumeration of states such that the 
probability of reaching states ``far to the right'' in the enumeration is small.
For general Markov processes, the state space contains a continuous component and their techniques do not apply.
Continuous components are necessary for $\deslang$ systems as we explicitly model time and continuous distributions over time.

Using machinery from the theory of topological recurrent Markov chains \cite{meyn2012markov}, \citet{MajumdarSS24} 
demonstrated a semi-completeness result for stochastic dynamical systems evolving over general 
state spaces endowed with topologies.
However, they required these systems to meet certain continuity assumptions over this topology that ensured that ``nearby'' states behaved ``similarly''.
More precisely, they assumed a \emph{weak Feller} property, i.e., 
that for any open set $O$ in the topology, $\liminf_{y \to x} T_X(y, O) \geq T_X(x, O)$ where $T_X(x, O)$ gives the 
probability of jumping from any state $x$ to~$O$.

Unfortunately, general Markov processes and, in  particular, Markov processes arising out of $\deslang$ systems, need not be weak Feller.
\smallCodeRef{fig:feller-ctx} describes an explicit counterexample.
After executing the initial procedure $p_0$, the execution reaches the configuration with queue $\set{(1, 0, a_1, v)}$, with $v$ sampled uniformly from $(0, 1)$.
A \textsc{dequeue} step deterministically empties the queue, sets the local variable $y$ to $v$, and sets the command $\com_1$ to the body of procedure $\mtt{p1}$.
Let $O$ be a small open set around $\goal = 1$, and take the configuration $\frakc = (\varnothing, S_{\msf{ctx}}, [y = 1/2], 1, \com_1)$ where the local variable $y = 1/2$. 
It is easy to see that
$\liminf_{\frakc' \to \frakc} T_O(\frakc') = 0$, as configurations $\frakc'$ with $y < 1/2$ (meaning $T_O(\frakc') = 0$) can be arbitrarily close to $\frakc$.
However, $T_O(\frakc) = 1 > 0$.

\begin{system}[t]
\centering
\begin{footnotesize}
\begin{minipage}{0.4\textwidth}
\begin{lstlisting}[mathescape=true]
proc $p_0$() = x $\sim$ $\uniform$(0, 1)
    post(p1, x, $\Dirac(1)$)
\end{lstlisting}
\end{minipage}
\begin{minipage}{0.55\textwidth}
\begin{lstlisting}[mathescape=true,firstnumber=3]
proc p1(y) = if (y < 1/2) then goal := 0
    else goal := 1
\end{lstlisting}

\end{minipage}
\end{footnotesize}
\caption{A $\deslang$ system $S_\msf{ctx}$ whose semantics is not weak-Feller.\label{fig:feller-ctx}}
\end{system}

Instead, the completeness of our proof rule can be established using recent advances in the theory of discrete-time Markov processes \cite{xu2018note}.
First, we produce an irreducible chain\footnote{
    For countable state spaces, a chain is irreducible if, for any pair of states, there is a finite path of non-zero probability linking the states. This notion is generalized using measure-theoretic constructions to apply in our setting.}
using the almost-sure reachability property; this gives us the invariant.
We then use insights from \citet{xu2018note} to produce a countable increasing cover over the state space $X$ such that each element in the cover satisfies a \emph{petiteness criterion}.\footnote{
    Informally, a set $C$ is \emph{petite} if there is a common measure on the state space that provides a uniform lower bound on the probability of reaching measurable subsets after a random number of steps, regardless of the starting point in $C$.
    } 
The key step uses Egorov's theorem to convert pointwise convergence into \emph{uniform convergence} on all but a negligible part of the
state space.
This argument removes the weak Feller requirement
from previous approaches \cite{MajumdarSS24,meyn2012markov}.
We then produce supermartingales using the structure of the cover using petite sets.
This division of the state space naturally induces a variant that marks a state using its position in the growing sequence.

The full proof, including measure-theoretic details, is in the supplementary material.
\end{proof}

\begin{system}[t]
\begin{footnotesize}

\begin{minipage}{0.6\textwidth}
\begin{lstlisting}[mathescape=true]
proc Walk() = u $\sim$ $\uniform(0, 1)$;
    if u $\leq$ $2^{-\lceil x \rceil}$ then x := $\normal$(x, 1);
    post(Walk, $\exponential(1.4)$)
\end{lstlisting}
\end{minipage}
\begin{minipage}{0.35\textwidth}
\begin{lstlisting}[mathescape=true,firstnumber=4]
proc $p_0$() = x := 1;
    post(Walk, $\Dirac(1)$)
\end{lstlisting}
\end{minipage}
\end{footnotesize}
\caption{The \emph{backoff walk} $S_\msf{bo}$: a symmetric random walk whose steps are attempted with a probability that decays geometrically in the current level.\label{fig:backoff-walk}}
\end{system}


\begin{example}
\label{ex:backoff-walk}
    Take the system $S_\msf{bo}$ of \codeRef{fig:backoff-walk}, a walk under \emph{exponential backoff}: from level $x$ a Gaussian step is attempted only with probability $q(x) = 2^{-\lceil x \rceil}$, and stalls (stays put) otherwise, inducing the kernel
    \[
      T_X(x, \cdot) \;=\; q(x)\,\normal(x, 1) \;+\; \bigl(1 - q(x)\bigr)\,\Dirac(x).
    \]
    We prove almost-sure reachability of $\Goal = \set{x \leq 0}$ with $\Inv = \reals$. Since $q$ is discontinuous at every integer, $S_\msf{bo}$ is not weak Feller (so, previous rules do not apply)
    and the Gaussian step is beyond $\pGCL$ and its continuous extensions \cite{McIverMorganBook,BatzKRW25}. 
    Reachability nonetheless holds, as the geometric number of stalls between attempts is almost surely finite.

    Take the exempt set $T = \set{x \leq 1}$ and the witnesses
    \[
    \spmtg(x) = x \cdot \indicator\set{x > 1},
    \qquad
    \vnt(x) = \bigl(x + \tfrac12\bigr) \cdot \indicator\set{x > 0},
    \]
    \[
    H(r) = \max(r, 1) + \tfrac12,
    \qquad
    d(r) = \tfrac12,
    \qquad
    \epsilon(r) = \Phi(-\tfrac12) \cdot \rev{2^{-\lceil r - 1/2 \rceil}},
    \]
    with $\phi$ and $\Phi$ denoting the standard normal density and cumulative distribution functions.
    For \textbf{V1}, a stall leaves $\spmtg$ unchanged, so outside $T$ its drift is that of the plain Gaussian step damped by $q$: $\Delta\spmtg(x) = -q(x)\bigl(x\,\Phi(1 - x) - \phi(1 - x)\bigr) \leq 0$ for $x > 1$, since $q > 0$ and, for $t = x - 1 \geq 0$, the Mills bound \cite{Birnbaum42,Mills26} gives $(1 + t)\,\Phi(-t) \geq \phi(t)$ (equality only as $t \downarrow 0$)---so $T$ is the tightest exempt set for this choice of $V$. 

    The variant shows why \textbf{V2} needs $\epsilon(r)$. A stall fixes $\vnt$ and an attempted step raises it half the time; it falls only on a step drawing below $-\tfrac12$, of probability $q(x)\,\Phi(-\tfrac12)$, which \emph{vanishes} as $x \to \infty$, so no constant $\epsilon$ works. \rev{But at states with $\vnt(x) \leq r$ we have $x \leq r - \tfrac12$, hence $q(x) \geq 2^{-\lceil r - 1/2 \rceil} > 0$ and the variant descends with probability $\epsilon(r)$.}
    The example also shows the utility of the exempt set: without it, we cannot easily describe a supermartingale that decreases in expectation in $(0, 1)$ and is zero on $\set{x\leq 0}$.
    \hfill\qed
\end{example}

\subsection{Almost-Sure Reachability: Proof Rule for $\deslang$}

In principle, one could use \cref{rule:des-ast-markov} to reason about the almost-sure reachability of $\deslang$ systems 
by applying it on the Markov process defining its semantics.
However, that would require building monolithic certificates over the space of all configurations of the $\deslang$ system.
We now provide a compositional proof rule for $\deslang$ systems.
The main difference from \cref{rule:des-ast-markov} is that, instead of a single supermartingale and and a single variant function, the rule asks for a supermartingale and variant for each procedure defined in the system. 
These witnesses are only evaluated at configurations where the command is the body of the procedure.
As we shall show later, we can compose these witnesses together to get global supermartingales and variants on the state space.

\paragraph{Notation.}
To describe our proof rule, we need some additional notation.
A tuple of an event queue, a program state, and a timestamp is called a \emph{full state}, denoted by $\scrs$.
The set of full states $\calQ \times \Sigma \times \timestamps$ is denoted by $\scrS$, where $\calQ$ is the set of event queues and $\Sigma$ is the set of program states.
Recall that configurations are denoted by $\frakc$ and $\frakc.Q$, $\frakc.S$, $\frakc.\sigma$, $\frakc.t$, and $\frakc.\sfc$ denotes its queue, $\deslang$ system, program state, time, and command, respectively.
Thus, full states are produced from configurations by dropping its command.
For a configuration $\frakc$, let $\frakc.\scrs$ denote the full state $(\frakc.Q, \frakc.\sigma, \frakc.t)$.
\rev{Full states carry the sigma-algebra of the sets $A \subseteq \scrS$ with $\set{\frakc \mid \frakc.\scrs \in A} \in \calB(\frakC)$, so a function $f$ on $\scrS$ is measurable exactly when $f(\frakc.\scrs)$ is measurable in $\frakc$.}
For functions $f$ over full states, we overload $f(\frakc)$ to mean $f(\frakc.\scrs)$.
Finally, when the context is clear, we let $\com_p$ denote the body of procedure $p$ defined in the $\deslang$ system $S$. 

Our proof rule for $\deslang$ systems decomposes \cref{rule:des-ast-markov} by asking for functions that defined on full states rather than configurations.
Because full states lack the command information, this new proof rule must ask for supermartingales and variants for each procedure defined in the system $S$.
For each procedure $p$ with body $\com_p$ in $S$, the proof rule uses the supermartingales and variants for $p$ only at command configurations with command $\com_p$.

Recall that, the operational semantics executes a $\deslang$ system by executing its loop-free procedures first, and then dequeues a procedure from the event queue for execution.
Thus, the proof rule must compare a procedure witness against the witness of the successor procedure reached after a dequeue in its supermartingale and variant conditions.

Our proof rule formalizes this using \emph{dequeue witnesses}.
Call a configuration $\frakc$ a \emph{command configuration} if $\frakc_\sfc \neq \sfskip$, a \emph{dequeue configuration} if $\frakc_\sfc = \sfskip$ and $\frakc_Q \neq \emptyset$, and a \emph{terminal} configuration otherwise.
The relation $\Rightarrow$ defining the operational semantics of $\deslang$ systems ensures that, from any command configuration, the successor configuration is almost surely a dequeue configuration.
The dequeue witnesses used in our proof rule are functions $\spmtg_\triangledown$ and $\vnt_\triangledown$ that affect a dequeue to produce the successor procedure witnesses $\spmtg_p$ and $\vnt_p$. 
For a full state $\scrs$, let $p_\triangledown(\scrs) = p$ and let $\scrs_\triangledown$ denote the full state $(\scrs.Q \setminus \min(\scrs.Q), \sigma[\bfy\mapsto\bfv], t_p)$,
where $\min(\scrs.Q) = (t_p, \_, p, \bfv)$. 
Then, $\spmtg_\triangledown(\scrs) \triangleq \spmtg_p(\scrs_\triangledown)$ and $\vnt_\triangledown \triangleq \vnt_p(\scrs_\triangledown)$.


\begin{figure}[h]
\begin{proofrule}[Almost-Sure Reachability for $\deslang$]
\label{rule:des-ast}
To show that a $\deslang$ $S$ almost surely reaches a measurable target $\Goal \subseteq \scrS$, find 
\begin{enumerate}[label=\arabic*.,leftmargin=1.5em]
    \item A \rev{measurable} \emph{invariant} $\Inv \supseteq \Goal$ containing the initial configuration $\frakc_0$,
    \item A \rev{measurable} \emph{exempt set} $T$ of full states with $\Goal \subseteq T \subseteq \Inv_\scrS$,
    \item \rev{Measurable} \emph{procedure supermartingales $\spmtg_p$ and variants $\vnt_p$} for each $p$ defined in $S$ mapping $\scrS$ to $\nnreals$,
    \item Assistant functions $H, d, \epsilon : \nnreals \to \preals$, 
\end{enumerate}
such that 
\begin{enumerate}[label=\Alph*.,leftmargin=1.5em]
    \item \emph{Witnesses vanish exactly on their target sets:} for each $p$ defined in $S$, $\spmtg_p(\fullStateArg) = 0$ \emph{iff} $\fullStateArg.\scrs \in T$, and $\vnt_p(\fullStateArg) = 0$ \emph{iff} $\fullStateArg.\scrs \in \Goal$,
\end{enumerate}
and for each $\frakc \in \Inv$ with $\frakc \Rightarrow \mu_\frakc$,
\begin{enumerate}[label=\Alph*.,leftmargin=1.5em,start=2]
    \item \emph{Supermartingale at procedures outside $T$:}
    if $\frakc.\sfc = \com_p$ for $p$ defined in $S$ and $\frakc.\scrs \notin T$, then
    $\spmtg_p(\fullStateArg) \geq \int_\frakC \spmtg_\triangledown(\fullStateArg') d\mu_\frakc(\frakc')$,
    \item \emph{Variant conditions:} for each $r \in \nnreals$,
    \begin{enumerate}[label=\alph*.,leftmargin=*]
        \item \emph{Variants are bounded:} if $\spmtg_p(\fullStateArg) \leq r$, then $\vnt_p(\fullStateArg) \leq H(r)$ for each $p$ defined in $S$,
        \item \emph{At procedures:}  if $\frakc.\sfc = \com_p$ for some $p$ defined in $S$, $\frakc.\scrs \notin \Goal$,
        and \rev{$\vnt_p(\fullStateArg) \leq r$}, then\\
        \hspace*{1em}
        $\mu_\frakc( \set{\frakc' \in \Inv \mid \vnt_\triangledown(\fullStateArg') \leq \vnt_p(\fullStateArg) - d(r)}) \geq \epsilon(r)$,
    \end{enumerate}
\end{enumerate}
\end{proofrule}

\caption{The proof rule for deducing the almost-sure reachability of $\deslang$ systems.}
\label{fig:rule-fast}
\end{figure}


Condition A in \cref{rule:des-ast} requires that each supermartingale $\spmtg_p$ is $0$ exactly on $T$ and each variant $\vnt_p$ is $0$ exactly on $\Goal$; in particular $\spmtg_p$ vanishes throughout $T$.
Condition B formally specifies the supermartingale property over procedures at states outside $T$.
Note that, after executing the body of a procedure, configurations with $\sfskip$ commands are reached almost surely.
The semantics of $\deslang$ is to dequeue these configurations deterministically.
$\spmtg_\triangledown$ performs this dequeue and outputs the value of the correct procedure supermartingale for condition B.
The conditions in C relate to the variant, and use the auxiliary functions $H$, $d$ and $\epsilon$ asked by the rule.

Condition C(a) means that the variants $\vnt$ are bounded from above by a function ($H$) of the supermartingale $\spmtg$.
Because $\spmtg \equiv 0$ on $T$, instantiating Condition C(a) at $r = 0$ forces $\vnt \leq H(0)$ throughout $T$; that is, the variant is automatically bounded on the exempt set, which is what lets it carry the system to $\Goal$ once $\spmtg$ has confined it to $T$.
Condition C(b) formally specifies the variant property at procedures outside $\Goal$: \rev{at states where $\vnt_p \leq r$, the variant decreases by at least $d(r)$ with probability at least $\epsilon(r)$}.


\begin{theorem}
\label{th:des-sound-and-complete}
    \Cref{rule:des-ast} is sound and complete.
\end{theorem}

To prove the soundness of \cref{rule:des-ast}, we show that certificates from the rule can be translated into certificates for \cref{rule:des-ast-markov}.
We prove completeness by translating the certificates of \cref{rule:des-ast-markov} to \cref{rule:des-ast}.
Let $S$ be a $\deslang$ system and let $(\frakC, T_\frakC)$ be its associated Markov process.

\paragraph{From $\deslang$ to Markov chains.}
Take a collection of certificates proving that a $\deslang$ system $S$ almost surely reaches a measurable set $\Goal$ of configurations.
Let $\Inv$ be the invariant for the $\deslang$ system. 
Define the restricted invariant $\Inv'$ as
\begin{equation}
\label{eq:reduced-inv}
\Inv' = \left\{\frakc \in \Inv \mid \frakc.\sfc = \sfskip \lor \left(\exists p \text{ is defined in } S \text{ and } \frakc.\sfc = \com_p \right)\right\}.
\end{equation}
Here, $P_S$ is the collection of procedures defined in $S$.
$\Inv'$ restricts $\Inv$ to configurations at which dequeues are imminent ($\frakc.\sfc = \sfskip$) or configurations reached immediately after dequeues ($\frakc.\sfc = S(p)$ for some $p \in S$).
Recall that the semantics of $\deslang$ commands implies that from command configurations, systems reach dequeue or terminal configurations almost surely.
Thus, $\Inv'$ is an invariant and trivially contains the initial configuration.
Because $\Goal$ is reached almost surely under $\Inv$ and $\Inv' \subseteq \Inv$, a measurable subset $G \subseteq \Goal$ also reached almost surely must be included in $\Inv'$.
Thus, the certificates restricted to $\Inv'$ still prove the almost-sure reachability of $\Goal$.
A global supermartingale $\spmtg$ and variant $\vnt$ can be constructed from the local witnesses $\spmtg_p$ and $\vnt_p$ using the functions $\spmtg_\triangledown$ and $\vnt_\triangledown$ as follows.
\begin{equation*}
\label{eq:combined_aligned}
\begin{aligned}
\spmtg(\frakc) & = \begin{cases}
\spmtg_p(\frakc.\scrs), & \frakc.\sfc = \com_p \\
\spmtg_\triangledown(\frakc.\scrs), & \frakc.\sfc = \sfskip \\
0, & \text{otherwise}
\end{cases}
& \quad \text{and} \quad\quad
\vnt(\frakc) & = \begin{cases}
\vnt_p(\frakc.\scrs), & \frakc.\sfc = \com_p \\ \vnt_\triangledown(\frakc.\scrs), & \frakc.\sfc = \sfskip \\
0, & \text{otherwise.}
\end{cases}
\end{aligned}
\end{equation*}
\rev{The configurations with a given command form a measurable set and the map $\scrs \mapsto \scrs_\triangledown$ is measurable, so $\spmtg$ and $\vnt$ are measurable whenever every $\spmtg_p$ and $\vnt_p$ is.}
Because $\spmtg_\triangledown$ and $\vnt_\triangledown$ already perform the deterministic dequeue step, the composed $\spmtg$ and $\vnt$ satisfy the Drift Criterion \textbf{V1} and Variant Criterion \textbf{V2} exactly where Conditions B and C hold for $S$---that is, outside the exempt set $T$, which carries across the translation unchanged---so if \cref{rule:des-ast-markov} is sound, then \cref{rule:des-ast} is sound.



\paragraph{From Markov chains to $\deslang$.}
Take a collection of certificates proving that the chain $(\frakC, T_\frakC)$ almost surely reaches $\Goal$.
We assume that $\Goal$ is constructed from a set $G$ of full states of $S$ such that $\Goal = \{\frakc \in \frakC \mid \frakc.\scrs \in G \}$.
Let $\Inv$ be the invariant for the Markov chain.
Again, define the reduced invariant $\Inv'$ from (\ref{eq:reduced-inv}); it will contain a measurable part of $\Goal$ and the initial state $\frakc_0$ of $S$.
Because $\Inv' \subseteq \Inv$, the properties of the witnesses $\spmtg$ and $\vnt$ apply over $\Inv'$.
The procedure and dequeue witnesses can be defined by reversing the prior translations from $\deslang$ to Markov chains.
\begin{align*}
V_p(Q, \sigma, t) &= V(Q, S, \sigma, t, S(p)) & U_p(Q, \sigma, t) &= U(Q, S, \sigma, t, S(p))
\end{align*}
Thus, completeness of \cref{rule:des-ast-markov} implies completeness of \cref{rule:des-ast}.
\hfill $\square$

\begin{figure}[t]
\begin{proofrule}[Expected Time to Reach the Goal for $\deslang$]
\label{rule:des-time}
To bound the expected time for a $\deslang$ $S$ to reach a measurable target $\Goal \subseteq \scrS$, first show that $S$ reaches $\Goal$ almost surely, and additionally find
\begin{enumerate}[label=\arabic*.,leftmargin=1.5em]
    \item A \rev{measurable} \emph{invariant} $\Inv \supseteq \Goal$ containing the initial configuration $\frakc_0$,
    \item For each $p \in P_S$, a \rev{measurable} \emph{supermartingale $W_p$ bounding the expected time to reach the goal set}, mapping $\scrS$ to $\nnreals \cup \set{\infty}$,
\end{enumerate}
inducing the dequeue witness $W_\triangledown(\scrs) \triangleq W_{p_\triangledown(\scrs)}(\scrs_\triangledown)$,
such that
\begin{enumerate}[label=\Alph*.,leftmargin=1.5em]
    \item \emph{Dominates the clock on the goal:} for each $p \in S$ and each full state $\scrs \in \Goal$, $W_p(\scrs) \geq \scrs.t$,
\end{enumerate}
and for each $\frakc \in \Inv$ with $\frakc \Rightarrow \mu_\frakc$,
\begin{enumerate}[label=\Alph*.,leftmargin=1.5em,start=2]
    \item \emph{Supermartingale at procedures:}
    if $\frakc.\sfc = S(p)$ for $p \in S$ and $\frakc.\scrs \notin \Goal$, then
    $W_p(\fullStateArg) \geq \int_\frakC W_\triangledown(\fullStateArg') d\mu_\frakc(\frakc')$.
\end{enumerate}
Then, the expected time to reach $\Goal$ is at most $W_{p_0}(\frakc_0)$.
\end{proofrule}

\caption{Bounding the expected time to reach the goal set.}
\label{fig:rule-time}
\end{figure}

\subsection{Bounding the Expected Reach Time}

Next, we give a proof rule for bounding the expected time to reach the goal in $\deslang$ systems.
Recall that time is a part of the $\deslang$ configuration, and is updated at each dequeue step.
Thus, time is effectively a program variable, and the expected reach time is therefore the expected value of the program variable $t$ representing the global clock at the first configuration with $\goal=1$.

The standard technique to establish upper bounds on program variables is to apply Park induction.
In our setting, this translates to a supermartingale $W$ that dominates the time at the goal, i.e., $W(\frakc) \geq \frakc.t$ for $\frakc \in \Goal$.
Intuitively, $W(\frakc)$ can be understood as the expected value of the time $t$ at the $\Goal$ from $\frakc$ plus $\frakc.t$.
As in \cref{rule:des-ast}, this function $W$ can also be decomposed into per-procedure functions $W_p$.
Condition A ensures the remainder is zero on $\Goal$ and Condition B forces it to absorb, in expectation, the passage of time at each dequeue.

Notice that \cref{rule:des-time} requires \cref{rule:des-ast}.
This is necessary for soundness.
If $\Goal$ is reached with probability strictly less than one, the true expected time to reach $\Goal$ is $+\infty$. 
However, finite $W$ satisfying the supermartingale conditions of \cref{rule:des-time} can exist because $W$ is no longer constrained by condition A in diverging runs.

\begin{theorem}
    \label{thm:des-time}
    \cref{rule:des-time} is sound and complete.
\end{theorem}


The soundness and completeness of this rule follows from expressing the expected time as a least fixed point and the computation of the fixed
point using Park induction \cite{Kozen85,HitchcockP72} (see, e.g., \citet[Theorem 5]{FengCSKKZ23}).
For soundness, the conditions make $W_p(\cdot)$ a pre-fixed point of the operation, and so it dominates the least fixed point.
For completeness, note that if $\Goal$ is reached almost surely, then the expected time to reach goal satisfies the constraints.

Note that one cannot bound the expected reach time of $\deslang$ systems by bounding the expected value of the $\Goal$'s hitting time $\sigma_{\Goal}$ in the Markov process $(\frakC, T_\frakC)$.
The symmetric random walk \codeRef{fig:backoff-walk} can be edited so that each \textsf{post} of \texttt{Walk} to $\Dirac(0)$ will induce $\bbE[\sigma_{\Goal}] = \infty$, but because these infinitely many steps happen at the same logical timestep, the expected time to reach the goal is trivially $0$.

\subsection{Exponential Distributions and Rate-Based Event Sources}
\label{subsec:post-rate}

While our rules are complete, finding explicit certificates can be tricky, especially if the program
configuration does not capture the right abstraction. Consider the following example.

\begin{figure}[t]
\centering
\footnotesize
\begin{minipage}[b]{0.40\linewidth}
  \centering
  \resizebox{\linewidth}{!}{%
  \begin{tikzpicture}[>=stealth,
    qn/.style={circle,draw=black!70,fill=black!10,thick,inner sep=1pt,
               minimum size=17pt,font=\small},
    goal/.style={qn,double,double distance=1.2pt},
    ev/.style={font=\scriptsize}]
    \node[goal] (n0) at (0,0)   {$0$};
    \node[qn]   (n1) at (1.5,0) {$1$};
    \node[qn]   (n2) at (3.0,0) {$2$};
    \node[qn]   (n3) at (4.5,0) {$3$};
    \node       (nd) at (5.6,0) {$\cdots$};
    \draw[->] (n0) to[bend left=45] node[ev,above]{$1$}     (n1);
    \draw[->] (n1) to[bend left=45] node[ev,above]{$p$}     (n2);
    \draw[->] (n2) to[bend left=45] node[ev,above]{$p$}     (n3);
    \draw[->] (n3) to[bend left=45] node[ev,above]{$p$}     (nd);
    \draw[->] (n1) to[bend left=45] node[ev,below]{$1{-}p$} (n0);
    \draw[->] (n2) to[bend left=45] node[ev,below]{$1{-}p$} (n1);
    \draw[->] (n3) to[bend left=45] node[ev,below]{$1{-}p$} (n2);
    \draw[->] (nd) to[bend left=45] node[ev,below]{$1{-}p$} (n3);
  \end{tikzpicture}}
  \\[4pt]
  (a) Classical M/M/1 queue.
\end{minipage}
\hfill
\begin{minipage}[b]{0.58\linewidth}
  \centering
  \resizebox{\linewidth}{!}{%
  \begin{tikzpicture}[>=stealth,
    rn/.style={rounded corners,draw=black!70,fill=black!10,thick,
               inner sep=2.5pt,font=\small},
    goal/.style={rn,double,double distance=1.2pt},
    sm/.style={font=\scriptsize},
    ev/.style={font=\scriptsize,inner sep=1.5pt,align=center},
    evw/.style={ev,fill=white}]
    \colorlet{condcol}{blue!60!black}\colorlet{probcol}{red!55!black}
    \node[rn]   (cq)   at (2.0,1.7)  {$C(n,\Delta t)$};
    \node[rn]   (cq1)  at (5.4,1.7)  {$C(n{+}1,\Delta t)$};
    \node       (ctop) at (6.7,1.7)  {$\cdots$};
    \node[goal] (g)    at (-0.5,0)   {$n=0$};
    \node[rn]   (sq)   at (2.0,0)    {$S(n,\Delta t)$};
    \node[rn]   (sq1)  at (5.4,0)    {$S(n{+}1,\Delta t)$};
    \node       (sbot) at (6.7,0)    {$\cdots$};
    \draw[->] (cq) to[bend left=15] node[ev,above]{\textcolor{probcol}{$1{-}e^{-\lambda \Delta t}$}} (cq1);
    \draw[->] (cq) to[bend right=8] node[evw,pos=0.27]{\textcolor{probcol}{$e^{-\lambda \Delta t}$}} (sq1);
    \draw[->] (sq1) to[bend left=15] node[ev,below]{\textcolor{probcol}{$1{-}e^{-\mu \Delta t}$}} (sq);
    \draw[->] (sq1) to[bend right=8] node[evw,pos=0.27,yshift=5pt]{\textcolor{probcol}{$e^{-\mu \Delta t}$}} (cq);
    \draw[->,dashed] (sq) -- node[ev,above,pos=0.5]{$n{=}1$} (g);
  \end{tikzpicture}}
  \\[4pt]
  (b) M/M/1 queue in $\deslang$.
\end{minipage}
\caption{M/M/1 queue, two views. \textbf{(a)} The classical model of the M/M/1 queue, a birth-death Markov chain on the queue length \(n\) (up w.p.\ \(p=\lambda/(\lambda+\mu)\), down w.p.\ \(1-p\), goal \(n=0\)), on which \(V=U=n\) is a certificate. \textbf{(b)} The same queue as a \(\deslang\) system. A configuration carries the residual delay \(\Delta t\) of the one already-scheduled event, giving the client-headed \(C(n,\Delta t)\) and server-headed \(S(n,\Delta t)\) states.
A service completion at \(n=1\) empties the queue and reaches the goal (dashed). With \(\Delta t\) in the state, \(n\) is no longer a supermartingale.}
\label{fig:mm1-des-chains}
\vspace*{-0.9em}
\end{figure}
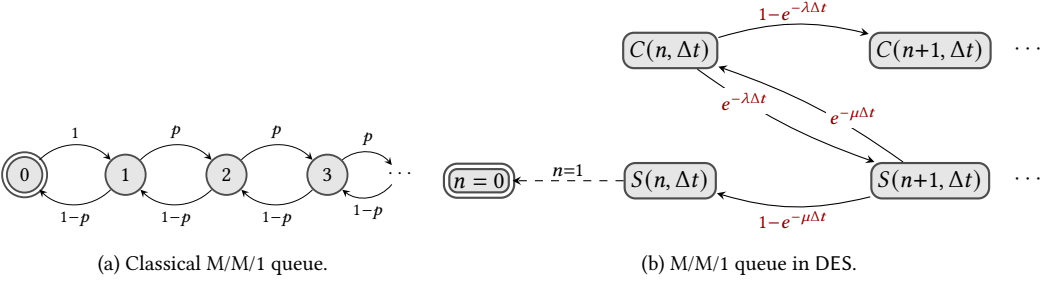

\begin{example}
Consider the client-server example from \cref{fig:mm1} that models an M/M/1 queueing system from queueing theory.
For this example, classical queueing theory answers the almost-sure reachability question: the queue empties almost surely iff the arrival
rate $\lambda$ is at most the service rate $\mu$.
The textbook argument models the system as a continuous-time Markov chain whose states correspond to the number of elements waiting
in the server's queue. 
With this representation, the queue size $n$ is a supermartingale and a variant. 
The supermartingale has non-positive drift $\lambda - \mu$ away from the empty state and the variant decreases with positive probability $\mu/(\lambda+\mu)$.

This certificate, however, does not discharge our proof rule, which requires a certificate on the \emph{configurations} of a $\deslang$ execution.
The illustration in \cref{fig:mm1-des-chains} contrasts the classical M/M/1 chain with its $\deslang$ execution.
Because $\deslang$ systems are headed by procedures, every state of the $\deslang$ M/M/1 is either a \textsf{Client} state or a \textsf{Server} state.
The \textsf{Client} deterministically increments the queue size $n$ and $\post$s a new client to the event queue with a delay sampled at the exponential rate $\lambda$.
Similarly, the \textsf{Server} decrements $n$ and $\post$s a new server with a delay sampled at rate $\mu$.
Thus, the only randomness at a state is the time at which the next event fires, determined by comparing the sampled delay with the residual time to the next pending event in the event queue, modeled as $\Delta t$ in \cref{fig:mm1-des-chains}.
At a server state, the server samples a fresh delay according to rate $\mu$.
With probability \(1-e^{-\mu\Delta t}\) the server's delay is shorter than the residual $\Delta t$, the next state remains a server state.
With the remaining \(e^{-\mu\Delta t}\) probability, the next state is a client state.
At the client, the transitions are symmetric with rate $\lambda$.
Thus, the arrival and service rates do not compete at the same configuration.

Thus, the queue length $n$ alone is not a supermartingale (in fact, $n$ increases deterministically at a \textsf{Client} state).
By our completeness results, certificates for the $\deslang$ M/M/1 can be obtained for $\lambda \leq \mu$.
But they cannot be pure functions of the queue length.
For example, when $\lambda < \mu$, we find 
\[
  V_C(n,t) = U_C(n,t) = n + A + \lambda t,
  \qquad
  V_S(n,t) = U_S(n,t) = n + D\,e^{-(\mu-\lambda)t},
\]
where \(D = \dfrac{\lambda}{\mu-\lambda}\) and \(A = 1 + \dfrac{\lambda^2}{\mu(\mu-\lambda)}\). 
The supermartingale holds with equality in the Client phase and with slack in the Server phase, and the
variant drops by a fixed amount on a race outcome of probability \(\Omega(e^{-r})\). 
When $\lambda=\mu$, the supermartingale can be expressed as the Snell envelope \cite{Snell52} of the process, but there is no simple closed form.
%
%
\hfill\qed
\end{example}

The reason why the classical model of M/M/1 queues is a continuous-time Markov chain (CTMC) is the memoryless property of the exponential distribution and the accompanying properties of racing exponentials.
In the CTMC model, when the queue is non-empty, the Server and the Client sample from their exponential distributions \emph{simultaneously}.
The lowest sample wins, and the winner procedure is executed with time moved forward.
By the property of racing exponentials, the value of this lowest sample is exponentially distributed with rate $\lambda + \mu$.
Then, by memorylessness, the loser's residual delay (the variable $\Delta t$ in \cref{fig:mm1-des-chains}) is probabilistically equivalent to a fresh exponential sample at the loser's rate.
This justifies the CTMC model of the M/M/1 system.


Given the importance of exponential distributions in simulations, we augment $\deslang$ with syntactic sugar that directly 
implements them.
This additional syntax (compiled into ghost state) makes many proofs simpler: for example, we recover the classical certificates from queueing theory.

\begin{figure}[t]
\begin{lstlisting}[mathescape=true,numbers=none,basicstyle=\scriptsize\ttfamily]
post_rate(p, $\lambda$, $\bfe$) $\defeq$ if $\rateSet = \emptyset$ then post(tick, $\Dirac$(0)); $\tid$ := $\tid$ + 1; $\rateSet$ := $\rateSet \cup \{(\lambda, \tid, p, \bfe)\}$

proc tick() =
  if $\rateSet \neq \emptyset$ then
    tp := peek();
    if tp = now() then post(tick, $\Dirac$(0))
    else
      te $\sim$ $\exponential$(R), (where R$=\sum_{(\lambda, \_, \_, \_)\in \Lambda} \lambda$);  ($\lambda$,id,q,v) $\sim$ pick($\rateSet$) each with prob $\lambda/R$;
      if now() + te < tp then
        $\rateSet$ := $\rateSet \setminus \{(\lambda,\text{id},q,v)\}$;  post(q, v, $\Dirac$(te));  post(tick, $\Dirac$(te))
      else post(tick, $\Dirac$(tp $-$ now()))
\end{lstlisting}
\caption{Compiling $\postrate$ into $\deslang$: a \ttt{post\_rate} command desugars (top) to maintain a global rate-source set $\rateSet$ and arm the administrative procedure $\tick$ (bottom), which runs the exponential race. 
}\label{fig:rateproc}
\end{figure}

We define a primitive $\postrate(p, \lambda, e_1, \ldots, e_k)$ that schedules a single execution of procedure $p$ on arguments $e_1, \ldots, e_k$ after a delay sampled from an exponential distribution with rate $\lambda$.
$\postrate$ captures the common idiom where a procedure is repeatedly fired at a rate given by an exponential distribution, for example, the arrival rate of clients and service rate of servers in \cref{fig:mm1}.

$\postrate$ is syntactic sugar.
\cref{fig:rateproc} rewrites any system using $\postrate$ into a plain $\deslang$ system that uses only $\post$ and the read-only primitives $\nowcmd$ and $\peekcmd$.
The compilation keeps the currently-active sources in an ordinary global variable $\rateSet$.
This variable contains a finite set of quadruples $(\lambda, \tid, p, \bfv)$ of a rate, a fresh identifier, a procedure, and its argument values.
Together with the procedure $\tick$, $\postrate$ simulates the exponential rates.
Observe that $\rateSet$ must hold several copies of the same source, just as the event queue $Q$ can hold several copies of the same event.
When $\tick$ fires (\cref{fig:rateproc}) it samples an aggregate delay $t_e \sim \exponential(R)$, where $R = \sum_{(\lambda,\_,\_,\_) \in \rateSet} \lambda$, together with a source, drawn with probability $\lambda/R$.
It fires that source if $t_e$ beats the next queued event, and otherwise re-arms, discarding the sample. 

At first glance $\postrate(p, \lambda, \bfe)$ resembles $\post(p, \bfe, \exponential(\lambda))$: they induce the same distribution over firing times.
The difference between $\postrate$ and $\post$ is apparent when several procedures are posted with exponential delays simultaneously, 
such as in the M/M/1 queue of \cref{fig:mm1-des-chains}.
A $\post$ is resolved with a sampling of the exponential to produce a delay that determines its point of entry into the event queue.
A $\postrate$ instead does not sample its delay immediately.
Instead, it stores the procedure in the program state, and inserts a special procedure $\tick$ (shown in \cref{fig:rateproc}) into the queue with $\Dirac(0)$ delay.
$\tick$ then assembles all exponential rates and runs an exponential race; if the winning delay beats the next pending event in the queue, it inserts only that winner into the event queue.
$\tick$ then posts itself to execute right after the next real procedure.
The remaining $\postrate$ procedures resample their delays in the exponential race the next time $\tick$ is executed.
In this way, $\postrate$ simulates precisely the CTMC modeling of exponential clocks.
For the M/M/1 queue, this means the residual delays in \cref{fig:mm1-des-chains} are no longer part of the configuration
and the queue size satisfies the drift criterion.

This compilation allows us to apply \cref{rule:des-ast} and \cref{rule:des-time} directly to the compiled program.
For any configuration $\frakc$ with the operational semantics giving $\frakc \Rightarrow \mu$, define the kernel
\[
  \tilde\mu_\frakc \;\triangleq\; \mu_\frakc \,\bind\, \tick,
\]
This kernel automatically sequences an execution of $\tick$ immediately after each one-step execution from $\frakc$.
(Essentially, this is the \textsc{concat} rule of \cref{fig:semantics-des}.)
A procedure's body may both post events directly (with $\post$) and arm rate sources (with $\postrate$); 
the next procedure is then whichever wins the race $\tick$ runs between the imminent committed $Q$-event and the active rate clocks.
Thus, using $\tilde\mu_\frakc$ instead of $\mu_\frakc$ in the conditions in \cref{rule:des-ast} and \cref{rule:des-time} gives sound and complete proof rules for reachability and expected reach time for programs using $\postrate$.

\begin{system}[t]
\begin{footnotesize}
\begin{minipage}{0.45\textwidth}
\begin{lstlisting}[mathescape=true]
proc Client() =
  req := ... ; // generate request
  $\textcolor{blue}{\mathtt{post}}$(SendReq, req, $\textcolor{brown}{\Dirac}(0)$);
  // arrivals at rate $\lambda$
  $\textcolor{red}{\mathtt{post\_rate}}$(Client, $\lambda$);

proc SendReq(req) =
  queue := enqueue(queue, req);
  $\textcolor{blue}{\mathtt{post}}$(Server, $\textcolor{brown}{\Dirac}(0)$);

proc $p_0$():
  queue, idle := [], true;
  $\textcolor{red}{\mathtt{post\_rate}}$(Client, $\lambda$);
\end{lstlisting}
\end{minipage}
\begin{minipage}{0.45\textwidth}
\begin{lstlisting}[mathescape=true,firstnumber=14]
proc Server() =
  if idle and not queue_empty then
    idle := false;
    req := dequeue(queue);
    $\textcolor{red}{\mathtt{post\_rate}}$(FinishReq, $\mu$, req);

proc FinishReq(req) =
  ... ; // service completion
  if not queue_empty then
    req := dequeue(queue);
    $\textcolor{red}{\mathtt{post\_rate}}$(FinishReq, $\mu$, req);
  else
    idle := true;
\end{lstlisting}
\end{minipage}
\end{footnotesize}
\caption{\codeRef{fig:mm1} with its exponential posts replaced by $\postrate$.\label{fig:mm1-rate}}
\end{system}

As an illustration, \codeRef{fig:mm1-rate} recasts the example of \codeRef{fig:mm1} using $\postrate$.
The model is otherwise unchanged.
Let us understand why \(\spmtg=\vnt=n\) discharges the rule in the modified model by reading the proof rule's conditions with \(\tilde\mu_\frakc\).
The only exponential sources in the rate-based model are the Client with rate \(\lambda\), and the server with rate \(\mu\).
\(\tick\) races them both, so from any busy configuration (\(n\ge 1\)) the next real procedure is the \textsf{Client} arrival with probability \(\lambda/(\lambda+\mu)\), sending \(n\mapsto n+1\), or the \textsf{Server} completion with probability \(\mu/(\lambda+\mu)\), sending \(n\mapsto n-1\).
This is exactly the process of the left in \cref{fig:mm1-des-chains}.
Thus, when $\lambda \leq \mu$, the drift of \(\spmtg=n\) is
\[ \int \spmtg_\triangledown\,d\tilde\mu_\frakc - \spmtg = \tfrac{\lambda}{\lambda+\mu}(+1) + \tfrac{\mu}{\lambda+\mu}(-1) = \tfrac{\lambda-\mu}{\lambda+\mu} \le 0 . \]


\section{Implementation and Case Studies}
\label{sec:case-study-queues}

\subsection{Mechanization in Lean}
\label{sec:impl}

We have implemented reasoning about $\deslang$ in our tool $\toolname$. $\toolname$ encodes a mechanization
of the $\deslang$ calculus, its rate-aware dequeue semantics, and the proof rules of \cref{rule:des-ast,rule:des-time} 
as a shallow embedding in the Lean~4 proof assistant, on top of \textsc{Mathlib}'s measure theory.
Every $\deslang$ system and certificate in this section is described explicitly in the embedding.
The user supplies a program in $\deslang$ and the certificates as input 
and  the embedding emits the proof obligations for our proof rules 
as \texttt{sorry}-ed goals, which we then discharge separately.

$\toolname$ implements a system-agnostic proof tactic, \flauto{}, that closes the emitted goals.
\flauto{} implements several search heuristics and can call SMT solvers. However, it is not always automatic and users may need to provide
system-specific lemmas or lemmas about probability distributions.
In our experience, a small number of lemmas are needed, on top of generic lemmas that we provide as part of the tool (such as properties of the
exponential distribution).
%
%
%
The completed development leaves no \texttt{sorry} and depends only on Lean's standard classical axioms.
All the following case studies have been formalized in $\toolname$.

Our choice of embedding in Lean was driven by three considerations: proofs involving measure theory can be subtle and we wanted the guarantees of the Lean checker;
proof obligations require mathematical reasoning about distributions that are not easily automatable using SMT; and 
the relatively mature search tactics (e.g., aesop) already available in Lean that discharges most obligations automatically.

\subsection{Benchmarks: Variations on Client-Server Systems}


\begin{table}[t]
\centering
{\small
\setlength{\tabcolsep}{5pt}
\renewcommand{\arraystretch}{1.35}
\begin{tabular}{>{\raggedright\arraybackslash}p{2.3cm}|c|c}
\textbf{Variation} & \textbf{Drift: \(V\), \(\Delta V\)} & \textbf{Variant: \(U\); \(d,\epsilon\)} \\
\hline
Basic
  & \parbox[c]{4.7cm}{\centering\(V = n\)\\[0.5mm]\(\Delta V = \lambda - \mu\)}
  & \parbox[c]{4.7cm}{\centering\(U = n\)\\[0.5mm]\(d = 1,\quad \epsilon = \tfrac{\mu}{\lambda+\mu}\)} \\
\hline
Feedback
  & \parbox[c]{4.7cm}{\centering\(V = n\)\\[0.5mm]\(\Delta V = \lambda - (1-p)\mu\)}
  & \parbox[c]{4.7cm}{\centering\(U = n\)\\[0.5mm]\(d = 1,\quad \epsilon = \tfrac{(1-p)\mu}{\lambda+\mu}\)} \\
\hline
Feedback with priority (\(k=3\))
  & \parbox[c]{4.7cm}{\centering\(V = n_1 + n_2 + n_3\)\\[0.5mm]\(\Delta V = \lambda - (1-p)\mu\)}
  & \parbox[c]{4.7cm}{\centering\(U = n_1 + 2n_2 + 3n_3\)\\[0.5mm]\(d = 1,\quad \epsilon = \tfrac{(1-p)\mu}{\lambda+\mu}\)} \\
\hline
Two servers in series
  & \parbox[c]{4.7cm}{\centering\(V = n_1^2 + n_2^2\)\\[0.5mm]\(\Delta V = 2n_1(\lambda{-}\mu_1) + 2n_2(\mu_1{-}\mu_2)\)\\[0.3mm]\({}+\,(\lambda{+}2\mu_1{+}\mu_2)\)}
  & \parbox[c]{4.7cm}{\centering\(U = 2n_1^2 + n_2^2\)\\[0.5mm]\(d = 1,\quad \epsilon = \tfrac{\min(\mu_1,\mu_2)}{\lambda+\mu_1+\mu_2}\)} \\
\hline
\(k=3\) servers in series
  & \parbox[c]{4.7cm}{\centering\(V = n_1^2 + n_2^2 + n_3^2\)\\[0.5mm]\(\Delta V = 2n_1(\lambda{-}\mu_1) + 2n_2(\mu_1{-}\mu_2)\)\\[0.3mm]\({}+\,2n_3(\mu_2{-}\mu_3) + (\lambda{+}2\mu_1{+}2\mu_2{+}\mu_3)\)}
  & \parbox[c]{4.7cm}{\centering\(U = 3n_1^2 + 2n_2^2 + n_3^2\)\\[0.5mm]\(d = 1,\quad \epsilon = \tfrac{\min_i \mu_i}{\lambda+\mu_1+\mu_2+\mu_3}\)} \\
\hline
Client with timeout
  & \parbox[c]{4.7cm}{\centering\rev{\(V = n + P + \tfrac{\kappa}{4}\bigl(Z - (1-\chi)\bigr)\)}\\[0.5mm]\(\Delta V \le 0\)\ {\footnotesize on \(\lambda(1{+}R)<\mu\)}}
  & \parbox[c]{4.7cm}{\centering\rev{\(U = n + (1+\theta^\star) P + \tfrac{\kappa}{4} Z\)}\\[0.5mm]\rev{\(d = \min(\tfrac{\kappa}{4},\, 1-\theta^\star)\)},\quad \(\epsilon = \tfrac{\mu}{\lambda+\mu}\)} \\
\end{tabular}
}
\caption{Supermartingale and variant functions for the client-server variations: the reachability supermartingale \(V\) with its drift \(\Delta V\), and the variant \(U\), which drops by \(d\) with probability \(\epsilon\).}
\label{tab:client-server-drift-variant}
\end{table}

Our first case study revisits the simple client-server example of \cref{sec:motivating-examples} through a family of variations, summarized in \Cref{tab:client-server-drift-variant}.
In each, the goal is the same: verify that every queue in the system eventually empties, from any initial configuration. For each variation the table lists a supermartingale \(V\) with its one-step drift \(\Delta V\) and a variant \(U\) that drops by at least \(d\) with probability at least \(\epsilon\), in the sense of \cref{rule:des-ast}.
Here, 
\(n\) ranges over the number of jobs in the queue. 
These global witnesses decompose into the local ones the rule asks for (\cref{subsec:rule-ast-markov}); their per-procedure forms, \rev{with the offsets that make each witness vanish exactly on its target set and the assistant function \(H\)}, the system code, and the machine-checked proofs are all provided in $\toolname$.

The \emph{Basic} example is the M/M/1 client-server system: requests arrive at rate \(\lambda\) and are served at rate \(\mu\).
Its certificate is \(V = U = n\), \rev{which} drains at rate \(\mu - \lambda\).
Its \emph{Feedback} variation re-enqueues a completed request with probability \(p\), capturing closed-loop client-server interactions as in web traffic and other service systems \cite{MorOpenVsClosed}. The $V = U = n$ certifies it, but the variant reduces with a slightly lower probability.

The \emph{Feedback with priority} variation maintains \(k\) priority queues\rev{: a new request enters the highest priority, and a finished request is re-enqueued at one priority lower with probability \(p\), completing after at most \(k\) levels}.
\rev{This models prioritizing new requests over long-running sessions, as when establishing new connections takes precedence over long-context LLM conversations.}
Certificates for reachability sum the queue sizes across the $k$ priority queues. The variant weights level \(i\) by its index, \(U = \sum_i i\,n_i\), so that a completing request lowers \(U\) by at least \(d = 1\), ordering the descent by priority.

The \emph{Two servers in series} and \emph{\(k\) servers in series} variations form a pipeline with service rates \(\mu_1, \dots, \mu_k\): a request enters at stage~1 and, on a stage-\(i\) completion, moves to stage \(i+1\), departing when the last server finishes.
The state records a per-stage occupancy \(n_i\) for each server.
The linear sum is no longer sufficient, so the supermartingale takes on the \emph{quadratic} form \(V = \sum_i n_i^2\). 
\rev{Its drift is negative outside a finite set, which the exempt set \(T\) absorbs; the diagonal form needs \(\lambda < \mu_1 < \dots < \mu_k\), and cross terms cover the whole stable region \(\lambda < \mu_i\).}
The variant \(U = \sum_i \beta_i n_i^2\) weights the stages upstream-first, \(\beta_i = k-i+1\), so advancing a job to a later stage strictly lowers \(U\); its descent probability is set by the slowest stage, \(\epsilon = \min_i \mu_i / (\lambda + \sum_i \mu_i)\).

Finally, \emph{Client with timeout \(\tau\) and \(R\) retries} models a client that sets a timeout of \(\tau\) units on each request.
If the request is not completed within \(\tau\), the client retries, up to \(R\) times.
A retry re-sends the request, and the server reprocesses every copy (there is no local caching), so the job queue may hold duplicate copies of a request.
The client stops retrying once it receives a response, tracked by a per-request flag; copies already enqueued are still served.
\rev{In this row $n$ counts the copies not yet served, $P$ the remaining retry budget, and $Z$ the pending deadlines, which both witnesses carry with the same weight.}
Its certificate augments the workload with a \emph{phase term} that tracks the imminent timeout, paired with a retry-budget variant.
Here \(\mtt{ret}\) is a request's used retries, and a request is \emph{live} while unanswered with budget remaining, so \(P = \sum_{\text{live}}(R-\mtt{ret})\) is the total remaining retry budget; \(\chi = e^{-R_\Lambda(t'-t)}\), with total rate \(R_\Lambda = \lambda+\mu\) and \(t, t'\) the current and next event times, measures how imminent the next timeout is; and the constants \(\kappa = \tfrac{\mu-\lambda(1+R)}{\lambda+\mu}\) and \(\theta^\star = \tfrac{\kappa(\lambda+\mu)}{\lambda+\mu+\lambda R}\) are positive exactly on the tight stability region \(\lambda(1+R) < \mu\).
The full system, certificates, and proofs are given in $\toolname$; together the two witnesses certify reachability and a finite expected time on \(\lambda(1+R) < \mu\).

\subsection{Case Study: Almost Sure Synchronization of Routing Protocols}

\begin{system}[t]
\centering
\begin{footnotesize}
\begin{minipage}{0.47\textwidth}
\begin{lstlisting}[mathescape=true]
proc r1_beginRead() =
  r1_toRead = true;
  if (not r1_sending) then
    r1_toRead = false;
    r1_reading = true;
    // perform read computation
    post(r1_finishRead, $\Dirac$($T_c$));

proc r1_finishRead() =
  r1_reading = false;
  if (r1_toSend) then
    // SYNCHRONIZED
    $\textcolor{red}{\goal}$ $\textcolor{red}{=}$ $\textcolor{red}{1;}$
    r1_toSend = false;
    r1_sending = true;
    // perform send computation
    post(r1_finishSend, $\Dirac$($T_c$));
  if (r1_toSchedule) then
    r1_toSchedule = false;
    // Restart routing protocol
    post(r1_beginSend,
      $\uniform$($T_p - T_r, T_p + T_r$));

proc $p_0$() =
  post(r1_beginSend, $\uniform$($0, T_p$));
  post(r2_beginSend, $\uniform$($0, T_p$));
\end{lstlisting}
\end{minipage}
\begin{minipage}{0.47\textwidth} 
\begin{lstlisting}[mathescape=true,firstnumber=27]
proc r1_beginSend() =
  r1_toSend = true;
  post($\textcolor{blue}{\textbf{r2}}$_beginRead, $\Dirac$(0));
  if (not r1_reading) then
    r1_toSend = false;
    r1_sending = true;
    // perform send computation
    post(r1_finishSend, $\Dirac$($T_c$));

proc r1_finishSend() =
  r1_sending = false;
  if (r1_toRead) then
    // SYNCHRONIZED
    $\textcolor{red}{\goal}$ $\textcolor{red}{=}$ $\textcolor{red}{1;}$
    r1_toRead = false;
    r1_reading = true;
    // perform read computation
    post(r1_finishRead, $\Dirac$($T_c$));
    r1_toSchedule = true;
  else
    // UNSYNCHRONIZED
    // Restart routing protocol
    post(r1_beginSend,
      $\uniform$($T_p - T_r, T_p + T_r$));
\end{lstlisting}
\end{minipage}
\end{footnotesize}
\caption{
    \emph{The Periodic Messages model.}
    We omit explicit descriptions of these procedures for router \ttt{r2} for brevity; there are procedures $\mtt{r2\_*}$ for each procedure $\mtt{r1\_*}$ that manipulate \ttt{r2} variables accordingly.
    Notice the \textcolor{blue}{blue} post of \ttt{r2\_beginRead}; in \ttt{r2\_beginSend}, this becomes \ttt{r1\_beginRead}.
    The posts of \ttt{r*\_beginSend} set the routing timer for router \ttt{r*}.
    Note that all booleans are initially \ttt{false}.
    \label{fig:routing-protocol}
}
\end{system}
\label{sec:case-study-2}

We now demonstrate the value of our techniques using a practical example from the literature on networked systems.
In 1993, Floyd and Jacobson \cite{floyd1993synchronization} discovered a pattern of periodic packet drops occurring near core routers at the New England Academic and Research Network.
These drops were thought to have been caused by routing updates, which are protocols used by routers to update information on network topology, such as RIP \cite{RIP} and IGRP \cite{IGRP}.
These protocols run at periodic intervals.
Routers executing these protocols are assumed to be independent; meaning that they all begin executing the protocol at different times.
Because these protocols are run periodically, network architects at the time assumed that the routers would remain unsynchronized.
However, Floyd and Jacobson found that routers executing any of a family of these protocols would eventually become \emph{synchronized}, meaning that they would all execute the protocol at the same time.
This synchronization causes delays through periodic packet drops, as routers across the network update their routing tables at the same time.

We will demonstrate this synchronization using our certificates.
We will model the general Periodic Messages model by Floyd and Jacobson that several routing protocols conform to \cite{RIP,IGRP} as a $\deslang$ system.
We note that Floyd and Jacobson augmented their findings with a discrete-time Markovian model to explore further the limits of this synchronization.
Their model makes strong assumptions on the time distances between different routing clusters. 
Our modeling of the Periodic Messages model makes none of these assumptions. 

We first denote the following constants:
\begin{inparaenum}
    \item $T_c$ is the computation time needed for a router to send or receive a message, 
    \item $T_p$ is the period of the routing protocol, and
    \item $T_r$ is a random jitter that is added to the period. 
\end{inparaenum}
It is assumed that $T_p \gg T_c$, and, for the synchronization certificate below, that $T_r < T_c/2$ (the
small-jitter regime; the large-jitter case is discussed as future work).

We describe a 2 router system in \smallCodeRef{fig:routing-protocol}.
In it, routers \ttt{r1} and \ttt{r2} are initialized at a uniformly picked time between $0$ and the period $T_p$.
Routers executing the protocol either send or receive messages. 
Routers require $T_c$ time units to process incoming and outgoing messages.
They begin the protocol by deciding to send messages pertaining to updates to the network, and end the protocol by scheduling the next send by (approximately) setting a routing timer to $T_p$.
In the middle of the protocol, they complete the send and process pending incoming updates.

We note a few features of this model.
\begin{inparaenum}
    \item Routers cannot send and receive messages simultaneously.
    Thus, when a router wishes to send a message when it is busy processing an incoming message, it must wait until the incoming message has been processed. This is because the information in the incoming messages may affect the content of the message the router wishes to send.
    \item When a router decides to send a message, the other router is immediately notified of an incoming message.
    This is to model negligible transmission delays between the routers.
    \item A uniform jitter of $T_r$ is introduced when setting the routing timer.
    A small $T_r$ models differences in the computation time for message processing, while a large $T_r$ models deliberate jitter included in the protocol specification.
\end{inparaenum}

We represent this model using procedures that define the actions of the routers at the beginning and ends of sends and receives.
At the beginnings, routers successfully schedule the ends when the routers are not already busy.
Notice that when beginning a send, the other router is immediately notified to begin a read.
The recipient router will choose to begin processing the incoming message if it isn't busy sending.
At the ends, the routers check and perform pending actions.
Critically, when ending sends, when the router attends to existing pending reads first before setting the routing timer to schedule the next execution of the protocol.



The two routers have \emph{synchronized} when they both begin executing the protocol within a $T_c$ unit time window.
When unsynchronized, a router begins by sending a routing message, which takes $T_c$ time.
This send is simultaneously processed by the receiving router.
The sending router then resets its timer and schedules its next execution.
However, when routers are synchronized, 
the receiving router will begin a send before it finishes processing its incoming message.
This forces the sending router to process this message before scheduling its routing timer.
Thus, because of the proximity of the timer expirations, both routers must additionally process incoming messages sent by the other before resetting the timer, adding an additional $T_c$ processing time for each.
Scaling this behaviour up from 2 to $N$ routers means that all synchronized routers will spend time $N \times T_c$ only processing routing messages.
This can cause delays, disrupting quality of service.

\subsubsection{Certifying Synchronization}
\label{subsubsec:floyd-jacobson-certs}

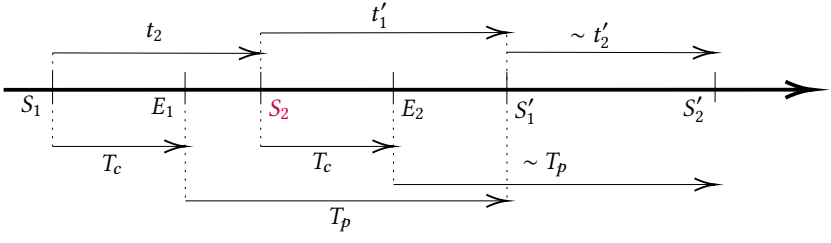
\begin{figure}
\centering
\scalebox{0.95}{

\begin{tikzpicture}[x=0.75pt,y=0.75pt,yscale=-1,xscale=1]

\draw [line width=1.5]    (94.03,120.13) -- (518.3,120.07) ;
\draw [shift={(521.3,120.07)}, rotate = 179.99] [color={rgb, 255:red, 0; green, 0; blue, 0 }  ][line width=1.5]    (14.21,-4.28) .. controls (9.04,-1.82) and (4.3,-0.39) .. (0,0) .. controls (4.3,0.39) and (9.04,1.82) .. (14.21,4.28)   ;
\draw    (119.9,110.68) -- (119.9,126.68) ;
\draw    (229.9,110.68) -- (229.9,126.68) ;
\draw    (359.9,110.68) -- (359.9,126.68) ;
\draw    (469.9,110.68) -- (469.9,126.68) ;
\draw  [dash pattern={on 0.84pt off 2.51pt}]  (120,100.68) -- (119.9,110.68) ;
\draw    (120,100.68) -- (227.28,100.93) ;
\draw [shift={(229.28,100.93)}, rotate = 180.13] [color={rgb, 255:red, 0; green, 0; blue, 0 }  ][line width=0.75]    (10.93,-3.29) .. controls (6.95,-1.4) and (3.31,-0.3) .. (0,0) .. controls (3.31,0.3) and (6.95,1.4) .. (10.93,3.29)   ;
\draw  [dash pattern={on 0.84pt off 2.51pt}]  (229.78,90.17) -- (229.9,110.68) ;
\draw    (229.78,90.17) -- (357.95,89.95) ;
\draw [shift={(359.95,89.95)}, rotate = 179.9] [color={rgb, 255:red, 0; green, 0; blue, 0 }  ][line width=0.75]    (10.93,-3.29) .. controls (6.95,-1.4) and (3.31,-0.3) .. (0,0) .. controls (3.31,0.3) and (6.95,1.4) .. (10.93,3.29)   ;
\draw    (189.9,110.68) -- (189.9,126.68) ;
\draw    (190.07,178.68) -- (357.97,178.68) ;
\draw [shift={(359.97,178.68)}, rotate = 180] [color={rgb, 255:red, 0; green, 0; blue, 0 }  ][line width=0.75]    (10.93,-3.29) .. controls (6.95,-1.4) and (3.31,-0.3) .. (0,0) .. controls (3.31,0.3) and (6.95,1.4) .. (10.93,3.29)   ;
\draw  [dash pattern={on 0.84pt off 2.51pt}]  (189.9,126.68) -- (190.07,178.68) ;
\draw    (119.85,149.98) -- (187.85,149.98) ;
\draw [shift={(189.85,149.98)}, rotate = 180] [color={rgb, 255:red, 0; green, 0; blue, 0 }  ][line width=0.75]    (10.93,-3.29) .. controls (6.95,-1.4) and (3.31,-0.3) .. (0,0) .. controls (3.31,0.3) and (6.95,1.4) .. (10.93,3.29)   ;
\draw  [dash pattern={on 0.84pt off 2.51pt}]  (119.9,126.68) -- (119.85,149.98) ;
\draw    (299.9,110.68) -- (299.9,126.68) ;
\draw  [dash pattern={on 0.84pt off 2.51pt}]  (229.9,126.68) -- (229.85,149.98) ;
\draw  [dash pattern={on 0.84pt off 2.51pt}]  (299.9,126.68) -- (300.05,170.08) ;
\draw    (229.85,149.98) -- (297.85,149.98) ;
\draw [shift={(299.85,149.98)}, rotate = 180] [color={rgb, 255:red, 0; green, 0; blue, 0 }  ][line width=0.75]    (10.93,-3.29) .. controls (6.95,-1.4) and (3.31,-0.3) .. (0,0) .. controls (3.31,0.3) and (6.95,1.4) .. (10.93,3.29)   ;
\draw    (300.05,170.08) -- (467.95,170.08) ;
\draw [shift={(469.95,170.08)}, rotate = 180] [color={rgb, 255:red, 0; green, 0; blue, 0 }  ][line width=0.75]    (10.93,-3.29) .. controls (6.95,-1.4) and (3.31,-0.3) .. (0,0) .. controls (3.31,0.3) and (6.95,1.4) .. (10.93,3.29)   ;
\draw  [dash pattern={on 0.84pt off 2.51pt}]  (359.78,91.17) -- (359.94,116.26) ;
\draw    (359.86,100.32) -- (467.14,100.56) ;
\draw [shift={(469.14,100.57)}, rotate = 180.13] [color={rgb, 255:red, 0; green, 0; blue, 0 }  ][line width=0.75]    (10.93,-3.29) .. controls (6.95,-1.4) and (3.31,-0.3) .. (0,0) .. controls (3.31,0.3) and (6.95,1.4) .. (10.93,3.29)   ;
\draw  [dash pattern={on 0.84pt off 2.51pt}]  (359.9,126.68) -- (359.97,178.68) ;

\draw (102.02,121.73) node [anchor=north west][inner sep=0.75pt]   [align=left] {$\displaystyle S_{1}$};
\draw (232.9,123.68) node [anchor=north west][inner sep=0.75pt]   [align=left] {\textcolor{purple}{$\displaystyle S_{2}$}};
\draw (362.9,123.68) node [anchor=north west][inner sep=0.75pt]   [align=left] {$\displaystyle S'_{1}$};
\draw (451.9,122.68) node [anchor=north west][inner sep=0.75pt]   [align=left] {$\displaystyle S'_{2}$};
\draw (167.9,83.68) node [anchor=north west][inner sep=0.75pt]   [align=left] {$\displaystyle t_{2}$};
\draw (286.9,72.68) node [anchor=north west][inner sep=0.75pt]   [align=left] {$\displaystyle t'_{1}$};
\draw (170.88,122.68) node [anchor=north west][inner sep=0.75pt]   [align=left] {$\displaystyle E_{1}$};
\draw (264.9,180.68) node [anchor=north west][inner sep=0.75pt]   [align=left] {$\displaystyle T_{p}$};
\draw (144.9,152.68) node [anchor=north west][inner sep=0.75pt]   [align=left] {$\displaystyle T_{c}$};
\draw (302.88,123.68) node [anchor=north west][inner sep=0.75pt]   [align=left] {$\displaystyle E_{2}$};
\draw (255.9,151.68) node [anchor=north west][inner sep=0.75pt]   [align=left] {$\displaystyle T_{c}$};
\draw (366.9,151.68) node [anchor=north west][inner sep=0.75pt]   [align=left] {$\displaystyle \sim T_{p}$};
\draw (391.9,83.68) node [anchor=north west][inner sep=0.75pt]   [align=left] {$\displaystyle \sim t'_{2}$};

\end{tikzpicture}

}
\caption{
    \emph{The timeline of unsynchronized execution.}
    Events $(S_1, E_1)$ and $(S_2, E_2)$ mark the points when \ttt{R1} and \ttt{R2} began and ended their executions.
    The current time is marked at \textcolor{purple}{$S_2$}.
    $S'_1$ marks the point when \ttt{R1}'s timer would expire.
    $S'_2$ marks the point when \ttt{R2}'s timer is \emph{expected to expire}.
    $t_2$ and $t'_1$ are the remaining time to reset the timers for routers \ttt{R2} and \ttt{R1}.
    $t'_2$ is the expected delay from the expiry of \ttt{R1}'s timer to \ttt{R2}'s timer.
}
\label{fig:timeline-floyd}
\end{figure}

We will demonstrate certificates proving that these routers synchronize almost-surely.
Accordingly, we set the distinguished variable $\goal$ to $1$ when a router must receive a message before processing a send or vice versa.
Our certificates will demonstrate the almost-sure reachability of this $\Goal$ set.
Because of our completeness results, our proof rule promises similar certificates to reason about the more complex case of $N$ routers.


It is intuitive to expect witnesses of synchronization to capture a notion of a ``time distance'' between the routers.
This time distance would measure the time gap between the beginnings of their protocol executions.
The protocol is designed to keep this distance constant in expectation.
However, the jitter $T_r$ when setting the routing timer causes it reduce with some probability at each restart.
Thus, this time distance is a suitable witness candidate for certifying reachability.

Let us examine this notion of distance in more detail.
Notice that, outside of synchronization, the protocol follows a fixed cyclic shape.
When router \ttt{r1} begins its protocol while router \ttt{r2} is idling, this distance corresponds to the remaining time $t_2$ in \ttt{r2}'s routing timer.
If $t_2 \mathrel{\rev{<}} T_c$, the routers are synchronized.
Assume $t_2 > T_c$, and trace the distance's evolution during execution.
After $T_c$ time units, \ttt{r1} completes its send and resets its routing timer to $T_p$ in expectation.
The execution then jumps forward to the expiration of \ttt{r2}'s timer.
Because $t_2 > T_c$ and, by protocol design, $t_2 \leq T_p$, \ttt{r1} is idle at this stage.
A similar derivation of the distance measure produces the remaining time in \ttt{r1}'s routing timer $t'_1 = T_p - (t_2 - T_c)$.
This variance is captured by \cref{fig:timeline-floyd}.

To address this variance, we always measure the time difference to the beginning of \ttt{r2}'s \emph{next} protocol execution and the beginning of \ttt{r1}'s protocol execution.
That is, we compare $t_2$, the initial remaining time on \texttt{r2}’s timer, with $t'_2$, the remaining time when \texttt{r1} begins its second execution.
As demonstrated in \cref{fig:timeline-floyd}, we estimate $t'_2$ as $T_p - (t'_1 - T_c) = t_2$ to be the value it will take in expectation.

To measure this time distance using the per-procedure certificates of \cref{rule:des-ast}, we must record the aforementioned cyclical behaviour of the protocol outside of synchronization.
We define an invariant to track this cycle 
and restrict the event queue to respect this cycle.
Outside of synchronization, one router leads (say \ttt{r1}), and the queue cycles deterministically through the entries
\ttt{r1\_beginSend}, \ttt{r2\_beginRead}, \ttt{r1\_finishSend}, \ttt{r2\_finishRead}, and then \ttt{r2\_beginSend} begins the mirror cycle with the roles swapped.
All other configurations are absorbed by $\Goal$ or by a deterministic path to synchronization once the time distance goes below $T_c$.

We use per-procedure supermartingales to measure this time distance. 
For a full state $\scrs$, let $t_2(\scrs)$ be the time $t_2$ at the event $(t_2, \_, \mtt{r2\_beginSend}) \in \scrs.Q$.
Similarly, $t_1(\scrs)$ is the time $t_1$ at the event $(t_1, \_, \mtt{r1\_beginSend}) \in \scrs.Q$.
For each full state $\scrs$ outside synchronization, define the functions
\begin{align*}
  \spmtg'_{\mtt{bs1}}(\scrs) &= t_2(\scrs) - \scrs.t 
  & \spmtg'_{\mtt{bs2}}(\scrs) &= T_p + T_c - (t_1(\scrs) - \scrs.t)\\ 
  \spmtg'_{\mtt{fs1}}(\scrs) &= t_2(\scrs) - \scrs.t + T_c 
  & \spmtg'_{\mtt{fs2}}(\scrs) &= T_p - (t_1(\scrs) - \scrs.t) 
\end{align*}
These functions represent the aforementioned distance measure.
\ttt{bs1}, \ttt{bs2}, \ttt{fs1}, and \ttt{fs2} correspond to the \ttt{beginSend} and \ttt{finishSend} procedures for routers \ttt{r1} and \ttt{r2} respectively.

Outside of the goal set, the \ttt{read} procedures are triggered by \ttt{send} procedures from other routers.
While they are performed at the same time as the sends (\ttt{r2\_beginRead} occurs immediately after \ttt{r1\_beginSend} at the same time stamp and \ttt{r2\_finishRead} occurs immediately after \ttt{r1\_finishSend}), the ``distances'' between the routers must be measured differently.
This is because the ``distance'' between the routers is measured by the distance between the \ttt{beginSend} events of both routers.
At \ttt{beginRead} events of router \ttt{r1}, the \ttt{beginSend} of router \ttt{r2} occurred at the same time.
Thus, set
\begin{align*}
  \spmtg'_{\mtt{br2}} &= \spmtg'_{\mtt{bs1}}
  & \spmtg'_{\mtt{br1}} &= \spmtg'_{\mtt{bs2}}
\end{align*}
However, at \ttt{finishRead} events of either router, outside of synchronization, two \ttt{beginSend} events must exist in the event queue.
The next \ttt{beginSend} to fire after a \ttt{r1\_finishRead} is \ttt{r1\_beginSend}, and after a \ttt{r2\_finishRead} it is \ttt{r2\_beginSend} (router $i$'s own timer expires before the peer's).
We therefore set each $\spmtg'_{\mtt{fr}i}$ to the value $\spmtg'_{\mtt{bs}i}$ takes when that \ttt{beginSend} fires --- that is, evaluated at $\scrs.t = t_i$, which is independent of $\scrs.t$:
\begin{align*}
  \spmtg'_{\mtt{fr1}}(\scrs) &= t_2(\scrs) - t_1(\scrs)
  & \spmtg'_{\mtt{fr2}}(\scrs) &= T_p + T_c - (t_1(\scrs) - t_2(\scrs))
\end{align*}
Here, \ttt{br1}, \ttt{br2}, \ttt{fr1}, and \ttt{fr2} correspond to the \ttt{beginRead} and \ttt{finishRead} procedures for routers \ttt{r1} and \ttt{r2} respectively.

For each procedure $p$, define
\begin{align*}
V_p(\scrs) = \begin{cases}
V'_p(\scrs) + 1, & V'_p(\scrs) \geq 0\\
1, & V'_p(\scrs) < 0 \land \scrs.\sigma.\goal \neq 1\\
0, & \scrs.\sigma.\goal = 1
\end{cases}
\end{align*}
Here $V'_p$ is the distance witness for the \emph{unsynchronized} region; under the assumption $T_r < T_c/2$ it
is non-negative on all reachable states, so the middle ($V'_p < 0$) case of $V_p$ is a defensive boundary value
that does not arise (the $+1$ keeps $V_p$ continuous at $V'_p = 0$).  Synchronization --- the $d \mathrel{\rev{<}} T_c$ and
collision states --- is instead handled by a bounded, deterministic dispatch sequence with $V_p = 1$ and a
matching variant.


Note that, because the $\spmtg_p$ supermartingales only decreases at the \ttt{finishSend} procedures, it is insufficient as a variant.
However, it can be transformed into a variant by scaling the supermartingale appropriately.
This defines the variants (for states with non-zero supermartingale values) as a descent-ladder offset
plus a
\emph{uniform} additive constant $T_r$:
\begin{align*}
  \vnt_{\mtt{bs1}} &= \spmtg_{\mtt{bs1}} + T_r/4 + T_r
  & \vnt_{\mtt{br2}} &= \spmtg_{\mtt{br2}} + T_r/8 + T_r
  & \vnt_{\mtt{fs1}} &= \spmtg_{\mtt{fs1}} + T_r \\
  \vnt_{\mtt{fr2}} &= \spmtg_{\mtt{fr2}} + 3T_r/8 + T_r
  & \vnt_{\mtt{bs2}} &= \spmtg_{\mtt{bs2}} + T_r/4 + T_r
  & \vnt_{\mtt{br1}} &= \spmtg_{\mtt{br1}} + T_r/8 + T_r\\
  \vnt_{\mtt{fs2}} &= \spmtg_{\mtt{fs2}} + T_r
  & \vnt_{\mtt{fr1}} &= \spmtg_{\mtt{fr1}} + 3T_r/8 + T_r
\end{align*}
It can be shown that the variants $\vnt$ decrease by $T_r / 8$ at each step with probability at least $1/4$.
This is because, at each \ttt{finishSend}, the routing timer is reset from $\uniform(T_p - T_r, T_p + T_r)$, so the now-independent supermartingales $\spmtg_{\mtt{fr1}}$ and $\spmtg_{\mtt{fr2}}$ reduce by $T_r / 2$ with that much probability.
In $\toolname$, we have mechanized this example with these certificates with a precise invariant, proving that synchronization is almost sure, validating the simulations of \cite{floyd1993synchronization}.

\section{Related Work}


\noindent
\textbf{Probabilistic program reasoning.}
A large body of work has studied logical and algorithmic reasoning principles for probabilistic programs \cite{KaminskiThesis,McIverMorganBook,ZilbersteinKST25,BaoDF25}.
Classical approaches generalize ranking functions to the probabilistic setting via
\emph{ranking supermartingales} \cite{mciver2005abstraction,ChakarovSankaranarayanan,FuC19,Chatterjee0020}. 
Subsequent work developed \emph{stochastic invariants}~\cite{chatterjee2017stochastic} and supermartingale-based
proof rules for almost-sure and quantitative termination~\cite{HuangFC18,ChatterjeeGMZ22}.
More recently, sound and relatively complete Hoare-style proof systems have been established for both qualitative and quantitative
termination of probabilistic programs~\cite{MajumdarS25,McIverMKK18}.
These systems, while foundational, address untimed probabilistic programs; they do not model
continuous time, event scheduling, or explicitly timed performance properties.

\smallskip
\noindent
\textbf{DES and performance modeling.}
DES is a well-established paradigm for modeling performance in computer
systems and networks. Practical frameworks such as OMNeT++~\cite{varga2008overview}, ns-3~\cite{riley2010ns},
and SimPy~\cite{matloff2008introduction} provide rich simulation constructs and efficient runtime engines, but
reasoning about correctness or performance properties is typically empirical, based on repeated
simulation and statistical analysis.
Complementary to these, formal modelling frameworks such as PEPA~\cite{PEPA} and Modest~\cite{ModestExtended,ModestPaper} provide compositional, semantically grounded approaches to describing stochastic and timed behaviour.
While these formalisms support analytical reasoning and quantitative verification, they target abstract finite-state models rather than general-purpose, event-driven simulation programs.
In contrast, classical queueing theory provides analytic reasoning techniques for
idealized models such as M/M/1 and M/M/$k$ queues~\cite{Kleinrock1975}.
However, these models assume fixed stochastic structures and do not support general compositional reasoning
over event-driven programs.

\smallskip
\noindent
\textbf{Formal methods for stochastic systems.}
Formal verification frameworks for stochastic and timed systems---including probabilistic model checking
and temporal logics such as PCTL and CSL---enable algorithmic analysis of Markov chains and
Markov decision processes~\cite{BaierKatoenBook}.
These approaches target finite-state transition systems and differ fundamentally from our
deductive proof system, which reasons about the operational semantics of a general
event-driven language.
To our knowledge, no prior work provides a sound and complete proof system
for performance and termination properties of discrete-event models.

\smallskip
\noindent
\textbf{Certificates for stochastic systems.}
Our approach draws on established proof rules for termination of probabilistic programs, in particular the use of ranking supermartingales and variant functions~\cite{McIverMorganBook,McIverMKK18,MajumdarS25}, together with classical drift conditions from Markov chain theory~\cite{Foster53,meyn2012markov}. These constructs provide quantitative measures of progress toward termination or target satisfaction. For linear stochastic systems, appropriate drift and variant functions can often be synthesized by solving linear matrix inequalities, while for polynomial dynamics, polynomial or SOS optimization can be employed to find suitable certificates. However, such functions need not always be semi-algebraic even when the underlying dynamics are polynomial~\cite{kordabad2025sum,kordabad2025certificates}.

Our work also connects to the rich literature on barrier-based reasoning in control and formal verification. Barrier certificates~\cite{PJP07,ames2019control,survey_automated} offer locally checkable sufficient conditions for safety properties and have been generalized to stochastic and hybrid settings~\cite{jagtap2020formal,abate2024stochastic}. Converse results establishing necessity of such certificates are comparatively rare, which either are developed for non-probabilistic systems \cite{prajna2005necessity,wisniewski2015converse} or raise restrictive assumptions \cite{xue2024sufficient}. In contrast, our framework derives such almost-sure guarantees by integrating variant-based progress reasoning with the semantics of discrete-event programs.

In addition, we note that certification methods for continuous-time Markov chains (CTMCs) remain relatively underexplored. For example, the sampling-based verification approach for uncertain parametric CTMCs \cite{Badings22CTMC} highlights the challenge of uncertain rates in CTMCs but does not directly provide certificate functions in the sense of ranking supermartingales or variant functions. Complementary to this, neural continuous-time supermartingale certificates \cite{Neustroev25NeuralCTSM} provide a general framework for synthesizing supermartingale-style certificates in continuous-time stochastic dynamical affine systems, including reachability  guarantees. However, completeness in all these cases remained open. Our framework builds on these ideas by offering explicit drift/variant style certificates for discrete-event and continuous-time stochastic systems under almost-sure reachability semantics, bridging classical Markov-chain drift conditions with variant-based progress reasoning.


\section{Conclusion}

Our paper takes an initial step in deductive verification for performance in an expressive programming model capturing the core 
features of discrete-event
simulators, including infinite state, continuous time, probabilistic choice, and event-driven execution.
We provide sound \emph{and complete} proof rules for almost sure reachability and expected reaching time for this language.
Our tool $\toolname$ is able to semi-automatically discharge the proof obligations for many interesting examples.
We are also able to formally analyze a famous
stochastic routing protocol from the networking literature.

\begin{acks}
This research was sponsored in part by
the Deutsche Forschungsgemeinschaft project 389792660 TRR 248--CPEC
(see \url{https://perspicuous-computing.science}).
The research of S. Soudjani is supported by the following grants: EIC 101070802, ERC 101089047, and EIC 101306368.
\end{acks}

\sloppy 

\label{beforebibliography}
\newoutputstream{pages}
\openoutputfile{main.pages.ctr}{pages}
\addtostream{pages}{\getpagerefnumber{beforebibliography}}
\closeoutputstream{pages}


\bibliography{references}                 

\label{afterbibliography}
\newoutputstream{pagesbib}
\openoutputfile{main.pagesbib.ctr}{pagesbib}
\addtostream{pagesbib}{\getpagerefnumber{afterbibliography}}
\closeoutputstream{pagesbib}

\appendix
\section{Soundness and Completeness Proofs}
\label{sec:proofs}

\subsection{Soundness}

Briefly, our techniques are sound for the following reasons.
The supermartingale $V$ is required to be a supermartingale outside the exempt set $T$, on which it vanishes.
\begin{revblock}
This forces the trajectories to visit some sublevel set of $V$ infinitely often almost surely.
Criterion \textbf{V2} gives, at each such visit, a non-zero lower bound on the probability of reaching $\Goal$, with the bound depending only on the sublevel set.
A zero-one law then yields the almost-sure reachability of $\Goal$.
\end{revblock}

We will now detail the above argument in full.
\begin{revblock}
\revkeep{Let $x_k, k \in \nats$ denote the state at the $k^{th}$ step of an execution of} $(X, T_X)$ from a state $x \in \Inv$, \revkeep{and let $\calF_k$ be the} sigma-algebra on the space of trajectories generated by $x_0, \ldots, x_k$.
We first show that, until the first visit to $\Goal$, the trajectory remains in the invariant.

\begin{lemma}
\label{lem:pre-goal-confinement}
For every $x \in \Inv$, the event
\[
J := \set{x_k \in \Inv \text{ for every } k \in \nats \text{ with } k \leq \sigma_\Goal}
\]
satisfies $\bbP_x(J) = 1$.
In particular, on the event $\set{\sigma_\Goal = \infty}$, the trajectory lies in $\Inv$ at every step almost surely.
\end{lemma}

\begin{proof}
For $k \in \nats$, let $J_k := \set{x_j \in \Inv \text{ for every } j \leq k \wedge \sigma_\Goal}$.
It is easy to see that each $J_k$ is $\calF_k$-measurable.
Additionally, $J_0$ holds because $x_0 = x \in \Inv$.
On the event $J_k \cap \set{\sigma_\Goal > k}$, the state $x_k$ lies in $\Inv \setminus \Goal$, so the invariant condition of \cref{rule:des-ast-markov} gives $\bbP_x(x_{k+1} \in \Inv \mid \calF_k) = T_X(x_k, \Inv) = 1$.
Thus, $J_{k+1} = J_k \cap \set{x_{k+1} \in \Inv}$ on this event, and on $\set{\sigma_\Goal \leq k}$, the events $J_{k+1}$ and $J_k$ coincide.
Hence $\bbP_x(J_{k+1}) = \bbP_x(J_k)$, making every $J_k$ almost sure by induction.
The event $J$ is the countable intersection $\bigcap_{k \in \nats} J_k$ of almost-sure events, and is therefore almost sure.
This completes the proof.
\end{proof}
\end{revblock}

We use the following notation for the level sets of the certificates.
We let, for any function $V : \Inv \to \nnreals$, $V_{\mathop{\bowtie} r}$ denote the set $\{x \in \Inv \mid V(x) \mathop{\bowtie} r\}$ for each $\bowtie\, \in \{\leq, <, \geq, >, =\}$.
We also let $V_{\in K}$ for any set $K \subseteq \nnreals$ denote $\{x \in \Inv \mid V(x) \in K\}$.
\rev{These sets are measurable when $V$ is.}

Recall that under \cref{rule:des-ast-markov}, the global supermartingale $V$ is a supermartingale only \emph{outside} the exempt set $T$, and $V \equiv 0$ on $T$ (the vanishing in criterion \textbf{V1}).
Let $\sigma_T = \inf\{k \in \nats : x_k \in T\}$ be the first time the chain enters $T$\rev{; since $\Goal \subseteq T$, $\sigma_T \leq \sigma_\Goal$}.
\rev{It is a random variable because $T$ is measurable.}
\rev{We next show that, on the event $\set{\sigma_\Goal = \infty}$, the trajectories almost surely visit some sublevel set of $V$ infinitely often.}


\begin{revblock}
\begin{lemma}
\label{lem:sublevel-io}
For every $x \in \Inv$,
\[
\bbP_x\left(\sigma_\Goal = \infty \text{ and } \forall m \in \nats : x_k \in V_{\leq m} \text{ for finitely many } k \in \nats\right) = 0.
\]
\end{lemma}

\begin{proof}
Let us first note the possibilities when $\sigma_\Goal = \infty$.
In this case, there are three possibilities regarding the number of times the trajectory visits the exempt set $T$.
If there are infinitely many visits, then $V(x_k) = 0$ for infinitely many $k$, and $V_{\leq 0}$ is visited infinitely often.
If there are finitely many and $\sup_k V(x_k) < \infty$, then the trajectory remains in $V_{\leq m}$ for $m := \lceil \sup_k V(x_k) \rceil \in \nats$, and $V_{\leq m}$ is visited infinitely often.
It thus suffices to show that
\begin{equation}
\label{eq:sublevel-obligation}
\bbP_x\left(\sigma_\Goal = \infty,\ x_k \in T \text{ for finitely many } k \in \nats,\ \sup_k V(x_k) = \infty\right) = 0.
\end{equation}

To show this, we first demonstrate a supermartingale process.
Consider the process $M_k := V(x_k) \mathbf{1}[\sigma_T \geq k]$.
Each $M_k$ is a random variable because $V$ is measurable.
Let $J = \set{x_j \in \Inv \text{ for every } j \in \nats \text{ with } j \leq \sigma_\Goal}$ be the $\bbP_x$-probability 1 event of \cref{lem:pre-goal-confinement} that the trajectory remains inside the invariant $\Inv$ until $\Goal$.
On $\set{\sigma_T \geq k + 1} \cap J$, we have $\sigma_\Goal \geq \sigma_T \geq k + 1$, so the state $x_k$ lies in $\Inv \setminus T$.
The drift criterion \textbf{V1} then gives $\bbE_x[V(x_{k+1}) \mid \calF_k] = \int_\Inv T_X(x_k, dy)\, V(y) \leq V(x_k) = M_k$.
Moreover, $M_{k+1} \leq V(x_{k+1})$, and outside $\set{\sigma_T \geq k + 1}$, $M_{k+1} = 0$.
Hence $\bbE_x[M_{k+1} \mid \calF_k] \leq M_k$ almost surely, \revkeep{so $M_k$ is a supermartingale of finite non-negative random variables, and the Martingale Convergence Theorem (Thm D.6.2 in \cite{meyn2012markov}) yields an almost surely finite} random variable \revkeep{$M_\infty$ with $M_k \to M_\infty$} almost surely.

\revkeep{On the event $\set{\sigma_T = \infty}$ we have} $M_k = V(x_k)$ for every $k$, so $V(x_k) \to M_\infty < \infty$, and the convergent sequence $V(x_0), V(x_1), \ldots$ is bounded in all but a $\bbP_{x_0}$-probability $0$ set of trajectories.
Thus, \revkeep{$\sup_k V(x_k) < \infty$ almost surely on} $\set{\sigma_T = \infty}$.
Because the analysis so far holds for any start state $y \in \Inv$, we have
\begin{equation}
\label{eq:no-T-bounded}
\bbP_y\left(\sigma_T = \infty \text{ and } \sup_{k} V(x_k) = \infty\right) = 0, \qquad \forall y \in \Inv.
\end{equation}

We now show \eqref{eq:sublevel-obligation}.
For $\ell \in \nats$, let $L_\ell := \set{x_\ell \in T \text{ and } x_j \notin T \text{ for all } j > \ell}$ be the event that the last visit to $T$ is at time $\ell$.
The event in \eqref{eq:sublevel-obligation} is the union over $\ell \in \nats$ of the events $\set{\sigma_\Goal = \infty} \cap L_\ell \cap \set{\sup_k V(x_k) = \infty}$, \revkeep{so it suffices to show} that each of these countably many events occurs with zero probability.
Fix $\ell \in \nats$.
On $J \cap \set{\sigma_\Goal = \infty} \cap L_\ell$, the state $x_{\ell+1}$ lies in $\Inv$.
On $L_\ell$, the values $V(x_0), \ldots, V(x_\ell)$ are finitely many finite numbers, so $\sup_k V(x_k) = \infty$ holds there precisely when $\sup_{k > \ell} V(x_k) = \infty$.
Therefore
\[
\begin{aligned}
&\set{\sigma_\Goal = \infty} \cap L_\ell \cap \set{\sup_k V(x_k) = \infty} \\
&\qquad \subseteq J^c \cup \left(\set{x_{\ell+1} \in \Inv} \cap \set{x_j \notin T \text{ for all } j > \ell} \cap \set{\sup_{k > \ell} V(x_k) = \infty}\right).
\end{aligned}
\]
The event $J^c$ has $\bbP_x$-probability $0$ by \cref{lem:pre-goal-confinement}.
For the second event, the Markov property at time $\ell + 1$ and \eqref{eq:no-T-bounded} give
\[
\begin{aligned}
&\bbP_x\left(x_{\ell+1} \in \Inv,\ x_j \notin T \text{ for all } j > \ell,\ \sup_{k > \ell} V(x_k) = \infty\right) \\
&\qquad = \bbE_x\left[\mathbf{1}[x_{\ell+1} \in \Inv] \times \bbP_{x_{\ell+1}}\left(\sigma_T = \infty \text{ and } \sup_k V(x_k) = \infty\right)\right] = 0.
\end{aligned}
\]
This proves \eqref{eq:sublevel-obligation} and completes the proof.
\end{proof}
\end{revblock}

\begin{revblock}
We next show that, from every state of a sublevel set of $V$, criterion \textbf{V2} yields a probability of reaching $\Goal$ that is bounded away from zero uniformly over the sublevel set.


\begin{lemma}
\label{lem:sublevel-reach}
Let $m \in \nats$, and set
\[
R := H(m), \qquad N(m) := \ceil{R / d(R)}, \qquad \eta(m) := \varepsilon(R)^{N(m)} > 0.
\]
Then $\bbP_x(\sigma_\Goal \leq N(m)) \geq \eta(m)$ for every $x \in V_{\leq m}$.
\end{lemma}

\begin{proof}
Write $\varepsilon := \varepsilon(R)$, $d := d(R)$, and $N := N(m)$.
Define the function $\ceil{U}(x) := \ceil{U(x) / d}$ on $\Inv$, so that $U(x) \leq R$ implies $\ceil{U}(x) \leq N$, and $\ceil{U}(x) = 0$ precisely when $x \in \Goal$.
We show by induction on $n \in \set{1, \ldots, N}$ that
\begin{equation}
\label{eq:descent-induction}
\bbP_x(\sigma_\Goal \leq n) \geq \varepsilon^n \qquad \text{for every } x \in \Inv \text{ with } U(x) \leq R \text{ and } 1 \leq \ceil{U}(x) \leq n.
\end{equation}
This proves the lemma.
Take $x \in V_{\leq m}$; criterion \textbf{V2} gives $U(x) \leq H(m) = R$, hence $\ceil{U}(x) \leq N$.
If $\ceil{U}(x) = 0$, then $x \in \Goal$ and $\sigma_\Goal = 0$.
Otherwise \eqref{eq:descent-induction} at $n = N$ gives $\bbP_x(\sigma_\Goal \leq N) \geq \varepsilon^N = \eta(m)$.

\paragraph{\revkeep{Base case:} $n = 1$.}
Take $x \in \Inv$ with $U(x) \leq R$ and $\ceil{U}(x) = 1$, that is, $0 < U(x) \leq d$; in particular $x \notin \Goal$.
The descent condition of criterion \textbf{V2} at the level $R$ \revkeep{implies that in one step, the chain reaches} $\set{y \in \Inv \mid U(y) \leq U(x) - d}$ with probability $\geq \varepsilon$.
Every state $y$ in this set has $U(y) \leq 0$, hence $y \in \Goal$.
Thus $\bbP_x(\sigma_\Goal \leq 1) \geq \varepsilon$.

\paragraph{\revkeep{Induction case.}}
Suppose \eqref{eq:descent-induction} holds for some $n < N$, and take $x \in \Inv$ with $U(x) \leq R$ and $\ceil{U}(x) = n + 1$; in particular $U(x) > d > 0$ and $x \notin \Goal$.
As in the base case, \revkeep{in one step the chain reaches} $D := \set{y \in \Inv \mid U(y) \leq U(x) - d}$ with probability $\geq \varepsilon$; since $T_X(x, D) \leq 1$, this also gives $\varepsilon \leq 1$.
The set $D$ is measurable because $U$ is, so $T_X(x, D)$ is defined.
Every $y \in D$ has $U(y) \leq U(x) - d \leq R$ and $\ceil{U}(y) \leq n$.
Either $y \in \Goal$, and $\bbP_y(\sigma_\Goal \leq n) = 1 \geq \varepsilon^n$, or $\ceil{U}(y) \geq 1$, and the induction hypothesis gives $\bbP_y(\sigma_\Goal \leq n) \geq \varepsilon^n$.
The Markov property at time $1$ gives
\[
\bbP_x(\sigma_\Goal \leq n + 1) \geq \int_D T_X(x, dy)\, \bbP_y(\sigma_\Goal \leq n) \geq \varepsilon^n \times T_X(x, D) \geq \varepsilon^{n + 1}.
\]
This completes the induction and the proof.
\end{proof}
\end{revblock}

\rev{The proof uses \textbf{V2} at a single level $R$ only: once $U(x) \leq R$, a descent from $x$ reaches a state $y$ with $U(y) \leq U(x) - d(R) \leq R$, so the condition applies at the same $R$ from $y$.}

\begin{revblock}
\revkeep{We first recall a zero-one law of probabilistic processes.}
\begin{lemma}[Zero-One Law of Probabilistic Processes {\cite[Lemma~2.6.1]{McIverMorganBook}}]
\label{lem:01}
Let $(A_i)_{i \geq 1}$ be events with $A_{i+1} \subseteq A_i$ for every $i \geq 1$, and let $\epsilon > 0$ satisfy $\bbP(A_{i+1}) \leq (1 - \epsilon) \times \bbP(A_i)$ for every $i \geq 1$.
Then $\bbP\left(\bigcap_{i \geq 1} A_i\right) = 0$.
\end{lemma}

\begin{proof}
\revkeep{The conditions of the lemma imply that} $\bbP(A_i) \leq (1 - \epsilon)^{i - 1} \times \bbP(A_1) \leq (1 - \epsilon)^{i - 1}$ for every $i \geq 1$.
Thus, $\sum_{i \geq 1} \bbP(A_i) \leq \frac{1}{\epsilon} < \infty$.
\revkeep{Applying the Borel-Cantelli lemma \cite{grimmett2020probability} to the above gives} $\bbP(A_i \text{ for infinitely many } i \geq 1) = 0$.
Since the events decrease, $\bigcap_{i \geq 1} A_i$ is contained in the event that $A_i$ occurs for infinitely many $i \geq 1$.
This completes the proof.
\end{proof}
In \cite{McIverMorganBook}, $A_i$ is the event that a process has not left a set of states $B$ within $i$ steps, and the hypothesis holds when the probability of leaving $B$ in one step is at least $\epsilon$ from every state of $B$.
\end{revblock}


\begin{revblock}
We now complete the proof of soundness of \cref{rule:des-ast-markov}.
For $m \in \nats$, let $E_m$ be the event that the trajectory visits $V_{\leq m}$ infinitely often.
By \cref{lem:sublevel-io}, the event $\set{\sigma_\Goal = \infty} \setminus \bigcup_{m \in \nats} E_m$ occurs with zero probability.
It thus suffices to show that $\bbP_{x_0}(E_m \cap \set{\sigma_\Goal = \infty}) = 0$ for every $m \in \nats$: the event $\set{\sigma_\Goal = \infty}$ is then contained in a countable union of events of $\bbP_{x_0}$-probability $0$.

Fix $m \in \nats$, and write $N := N(m)$ and $\eta := \eta(m)$ for the constants of \cref{lem:sublevel-reach}.
Let $\tau_1 < \tau_2 < \cdots$ be the visit times to $V_{\leq m}$ that are at least $N$ steps apart,
\[
\tau_1 := \inf\set{k \in \nats \mid x_k \in V_{\leq m}}, \qquad \tau_{i+1} := \inf\set{k \geq \tau_i + N \mid x_k \in V_{\leq m}}, \qquad i \geq 1,
\]
with $\inf \emptyset := \infty$.
For $i \geq 1$, let $A_i := \set{\tau_i < \infty \text{ and } \sigma_\Goal > \tau_i + N}$ be the event that the $i$-th visit occurs and $\Goal$ is not reached within $N$ steps of it.
These events decrease, $A_{i+1} \subseteq A_i$, because $\tau_{i+1} \geq \tau_i + N$.
Moreover, $E_m \cap \set{\sigma_\Goal = \infty} \subseteq \bigcap_{i \geq 1} A_i$, because on $E_m$ every $\tau_i$ is finite.
We show that $\bbP_{x_0}(A_{i+1}) \leq (1 - \eta) \times \bbP_{x_0}(A_i)$ for every $i \geq 1$.
\cref{lem:01} then gives $\bbP_{x_0}\left(\bigcap_{i \geq 1} A_i\right) = 0$, and hence $\bbP_{x_0}(E_m \cap \set{\sigma_\Goal = \infty}) = 0$.
On $\set{\tau_{i+1} < \infty}$, the state $x_{\tau_{i+1}}$ lies in $V_{\leq m}$, so \cref{lem:sublevel-reach} gives $\bbP_{x_{\tau_{i+1}}}(\sigma_\Goal > N) \leq 1 - \eta$.
The strong Markov property (see \cite[Section~3.4]{meyn2012markov}) at $\tau_{i+1}$ therefore gives
\[
\begin{aligned}
\bbP_{x_0}(A_{i+1}) &= \bbE_{x_0}\left[\mathbf{1}[\tau_{i+1} < \infty,\ \sigma_\Goal > \tau_{i+1}] \times \bbP_{x_{\tau_{i+1}}}\left(\sigma_\Goal > N\right)\right] \\
&\leq (1 - \eta) \times \bbP_{x_0}\left(\tau_{i+1} < \infty,\ \sigma_\Goal > \tau_{i+1}\right) \leq (1 - \eta) \times \bbP_{x_0}(A_i),
\end{aligned}
\]
where the last step holds because $\tau_{i+1} \geq \tau_i + N$ places $\set{\tau_{i+1} < \infty,\ \sigma_\Goal > \tau_{i+1}}$ inside $A_i$.

Because the argument uses no property of $x_0$ beyond membership in $\Inv$, we have $\bbP_x(\sigma_\Goal < \infty) = 1$ for every $x \in \Inv$, completing the proof.
\end{revblock}

\begin{figure}[t]
\centering
\resizebox{0.96\linewidth}{!}{%
\begin{tikzpicture}[>=stealth,
  qn/.style={circle,draw=black!70,thick,inner sep=1pt,minimum size=17pt,font=\small},
  goal/.style={qn,double,double distance=1.2pt},
  ev/.style={font=\scriptsize},
  row/.style={font=\scriptsize}]
  \node[qn] (n0) at (0,0)   {$0$};
  \node[qn] (n1) at (1.5,0) {$1$};
  \node[qn] (n2) at (3.0,0) {$2$};
  \node[qn] (n3) at (4.5,0) {$3$};
  \node[qn] (n4) at (6.0,0) {$4$};
  \node     (nd) at (7.2,0) {$\cdots$};
  \node[goal] (g) at (9.2,0) {$g$};
  \draw[->] (n0) to[bend left=45] node[ev,above]{$1$} (n1);
  \draw[->] (n1) to[bend left=45] node[ev,above]{$1/2$} (n2);
  \draw[->,very thick] (n2) to[bend left=45] node[ev,above]{$1/2$} (n3);
  \draw[->] (n3) to[bend left=45] node[ev,above]{$1/2$} (n4);
  \draw[->] (n4) to[bend left=45] node[ev,above]{$1/2$} (nd);
  \draw[->] (n1) to[bend left=45] node[ev,below]{$1/2$} (n0);
  \draw[->] (n2) to[bend left=45] node[ev,below]{$1/2$} (n1);
  \draw[->] (n3) to[bend left=45] node[ev,below]{$1/2$} (n2);
  \draw[->] (n4) to[bend left=45] node[ev,below]{$1/2$} (n3);
  \draw[->] (nd) to[bend left=45] node[ev,below]{$1/2$} (n4);
  \draw[->] (g) to[out=60,in=120,looseness=8] node[ev,above]{$1$} (g);
  \node[draw,dashed,rounded corners,fill=black!10,opacity=0.4,inner sep=20pt,fit=(n0)(n2),label={[font=\scriptsize]above:$V_{\leq 2}$}] {};
  \node[row] at (-1.1,-1.3)  {$V$};
  \node[row] at (-1.1,-1.75) {$U$};
  \foreach \x/\v/\u in {0/0/1, 1.5/1/\tfrac12, 3.0/2/\tfrac14, 4.5/3/\tfrac18, 6.0/4/\tfrac{1}{16}} {
    \node[row] at (\x,-1.3)  {$\v$};
    \node[row] at (\x,-1.75) {$\u$};
  }
  \node[row] at (9.2,-1.3)  {$0$};
  \node[row] at (9.2,-1.75) {$0$};
\end{tikzpicture}}
\caption{\rev{The symmetric random walk on $\nats$ with the target $g$, which no transition enters and which the chain never leaves, with $V$ and $U$ below each state and $T = \set{g, 0}$.
The dashed box is the sublevel set $V_{\leq 2}$; the descent it certifies from $2$ (thick) leaves it.}}
\label{fig:v-indexed-counterexample}
\end{figure}
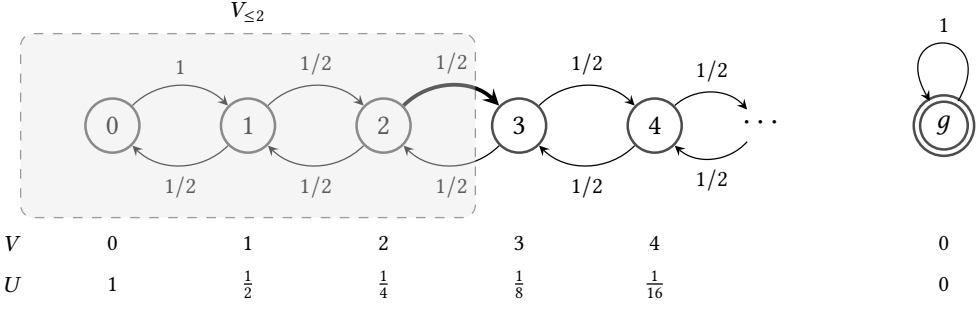

\begin{revblock}
\paragraph{Why \citet{MajumdarSS24} is unsound.}
The rule of \cite{MajumdarSS24} is \cref{rule:des-ast-markov} with the descent condition of criterion \textbf{V2} with $V(x) \leq r$ in place of $U(x) \leq r$.
It certifies the symmetric random walk of \cref{fig:v-indexed-counterexample} on $X = \nats \cup \set{g}$, which never reaches $\Goal = \set{g}$ from $x_0 = 0$: take $\Inv := X$, $T := \set{g, 0}$, $V(g) := U(g) := 0$, and
\[
V(x) := x, \qquad U(x) := 2^{-x}, \qquad x \in \nats; \qquad\qquad H(r) := 1, \qquad d(r) := 2^{-(r + 1)}, \qquad \varepsilon(r) := \tfrac12.
\]
$V$ vanishes exactly on $T$ and has zero drift at every $x \geq 1$; $U$ vanishes exactly on $\Goal$ and is at most $1 = H(r)$; and from every $x \in \nats$ with $V(x) \leq r$, the move to $x + 1$, of probability at least $\tfrac12$, lowers $U$ by $2^{-(x + 1)} \geq d(r)$.

The rule is unsound because a certified descent may leave the sublevel set of $V$ that certified it.
When this happens, the next sublevel set may certify a decrease of $U$ only by a smaller $d$.
The walk above leaves every $V_{\leq r}$; the certified decreases $2^{-(x + 1)}$ of $U$ shrink along the run and sum to at most $U(x_0) = 1$, so $U$ never reaches $0$.

Note that this certificate does not meet criterion \textbf{V2} of \cref{rule:des-ast-markov}: every $x \in \nats$ has $U(x) \leq 1$, and no $d(1) > 0$ lies below $2^{-(x + 1)}$ for all $x$.
\end{revblock}

\subsection{Completeness}

Broadly, our proof of completeness uses techniques from the literature on Foster-Lyapunov criteria for recurrent Markov chains \citet{meyn2012markov,xu2018note}.
To use such techniques however, we must first transform our chain such that these techniques can apply.

Let $(X, T_X)$ be a Markov chain. 
Let $x_0$ be the initial state of the chain, and let $\Goal$ be a measurable subset of $X$.
Suppose it is known that $\bbP_{x_0}(\sigma_\Goal < \infty) = 1$.

\begin{lemma}\label{lem:measurable_inv}
The set
\[
\Inv:=\{x\in X \mid P_x(\sigma_{\Goal}<\infty)=1\}
\]
is measurable.
\end{lemma}

\begin{proof}
Define $u:X\rightarrow[0,1]$ with
\[
u(x):=P_x(\sigma_{\Goal}<\infty), \qquad x\in X.
\]
Since $\Goal$ is measurable and $T_1$ is a Markov kernel, the finite-horizon hitting probabilities $u_n:X\rightarrow [0,1]$ defined with
\[
u_n(x):=P_x(\sigma_{\Goal}\le n)
\]
are measurable for every $n\ge0$. Moreover,
\[
u(x)=\lim_{n\to\infty}u_n(x).
\]
Therefore, $u$ is measurable as the pointwise limit of measurable functions, which makes
$\Inv
=
u^{-1}(\{1\})
$
a measurable set. This is due to being the inverse image of $\{1\}$, which is a Borel subset of $[0,1]$.
\end{proof}

Set a new transition kernel $T_1(x) := T_X(x)$ for all $x \not\in \Goal$ and $T_1(x_g, \{x_0\}) := 1$ for all $x_g \in \Goal$.
Define $\Inv \triangleq \{ x \in X \mid \bbP_{x}(\sigma_\Goal < \infty) = 1 \}$ with $\bbP_{x}$ measured from the kernel $T_1$.

Now, if $T_1(x, \Inv) < 1$ for some $x \in \Inv$, then $\bbP_{x}(\sigma_\Goal = \infty) = \int_{X \setminus \Inv} T_1(x, dy) \bbP_{y}(\sigma_\Goal = \infty) > 0$, which contradicts $\bbP_{x}(\sigma_\Goal < \infty) = 1$.
Thus, $T_1(x, \Inv) = 1$ for all $x \in \Inv$.
This lets us define a Markov kernel $T_\Inv(x, A) \triangleq T_1(x, A \cap \Inv)$ for each $x \in \Inv$ and $A \in \calB(X)$, giving us the restricted Markov chain $(\Inv, T_\Inv)$, on which every state reaches $\Goal$ almost surely.
This is the invariant, and the chain we work with for the rest of the proof.

We first show that
\[
\int_{X \setminus \Inv} T_1(x,dy)\,
\bbP_y(\sigma_{\Goal}=\infty)>0.
\]
To this end, we invoke Lemma~\ref{lem:measurable_inv}, which ensures that the set
$
\Inv=\{x\in X:\bbP_x(\sigma_{\Goal}<\infty)=1\}
$
is measurable. Consequently,
$
X\setminus\Inv
=
\{y\in X:\bbP_y(\sigma_{\Goal}=\infty)>0\}
$
is also measurable. Moreover, the mapping
$
y\mapsto\bbP_y(\sigma_{\Goal}=\infty)
$
is measurable, and
$
\bbP_y(\sigma_{\Goal}=\infty)>0,
\forall\,y\in X\setminus\Inv.
$
Therefore, if
$
T_1(x,X\setminus\Inv)>0,
$
then, by the positivity of the integral of a nonnegative measurable function,
\[
\int_{X\setminus\Inv}
\bbP_y(\sigma_{\Goal}=\infty)\,
T_1(x,dy)
>0.
\]

We now introduce the notion of $\psi$-irreducibility for the chain $(\Inv, T_\Inv)$.
A Markov chain $(\Inv, T_\Inv)$ is said to be \emph{\(\psi\)-irreducible} if there exists a nontrivial $\sigma$-finite measure \(\psi\) on \((\Inv, \calB(\Inv))\) such that for every measurable set \(A \in \calB(\Inv)\) satisfying \(\psi(A) > 0\),
\[
\forall x \in \Inv, \quad \bbP_x(\sigma_A < \infty) > 0. 
\]

Let \(T_\Inv^n\) denote the \(n\)-step transition kernel defined recursively by
\[
T_\Inv^{n+1}(x, A) = \int_\Inv T_\Inv(y, A) \, T_\Inv^n(x, dy),
\quad T_\Inv^1(x, A) = T_\Inv(x, A).
\]
Define the measure
$$
\psi(A) = \sum_{n \in \nats} 2^{-n} \times T_\Inv^n(x_0, A).
$$
It is easy to see that for all $x \in \Inv$ and $A \in \calB(\Inv)$ with $\psi(A) > 0$,
$$
\bbP_x(\sigma_A < \infty) \geq \bbP_x(\sigma_\Goal < \infty) \times \bbP_{x_0}(\sigma_A < \infty) > 0
$$
Thus, $(\Inv, T_\Inv)$ is $\psi$-irreducible with this measure $\psi$.

We must now introduce the notion of the \emph{petite set}.
Towards this, we first introduce the sampled chains.
A \emph{sampling distribution} $a$ is a probability distribution over $\nats$.
For a Markov process $(\Inv, T_\Inv)$, declare a transition kernel $K_a : x \times \calB(\Inv) \to [0, 1]$ as $K_a(x, A) = \sum_{n = 0}^\infty a(n) \times T_\Inv^n(x, A)$.
The process $(\Inv, K_a)$ is the \emph{sampled chain} of $(\Inv, T_\Inv)$ according to the sampling distribution $a$.
A measure $\nu$ over the measure space $(\Inv, \calB(\Inv))$ is said to non-trivial if it assigns a positive value to at least one set in $\calB(\Inv)$.
A set $C \subseteq \Inv$ is called \emph{petite} if there exists a sampling distribution $a$ and a non-trivial measure $\nu_a$ such that $K_a(x, B) \geq \nu_a(B)$ for all $x \in C$ and $B \in \calB(\Inv)$.

We will now show that $\Goal$ is petite in the chain $(\Inv, T_\Inv)$.
Take the sampling distribution $a$ with $a(1) = 1$ and $a(n) = 0$ for $n \neq 1$, such that $K_a(x, \cdot) = T_\Inv(x, \cdot)$.
Because every goal state moves to $x_0$ in one step, $T_\Inv(x, \cdot) = \delta_{x_0}$ for all $x \in \Goal$.
Hence, with the non-trivial measure $\nu = \delta_{x_0}$ (the point mass at $x_0$),
$$
K_a(x, B) = T_\Inv(x, B) = \delta_{x_0}(B) = \nu(B), \quad \forall x \in \Goal,\ B \in \calB(\Inv),
$$
so in particular $K_a(x, B) \geq \nu(B), \forall x \in \Goal$.
Thus, $\Goal$ is petite.

At this stage, we can show that the chain $(\Inv, T_\Inv)$ is \emph{recurrent}.
For $A \in \calB(\Inv)$, let
$$
U(x, A) := \sum_{n = 1}^\infty T_\Inv^n(x, A) = \bbE_x\!\left[\sum_{n \geq 1} \mathbf{1}[x_n \in A]\right]
$$
be the expected number of visits the chain makes to $A$ starting from $x$, and write $\calB^+(\Inv) := \{A \in \calB(\Inv) : \psi(A) > 0\}$.
A $\psi$-irreducible chain is \emph{recurrent} if $U(x, A) = \infty$ for every $x \in \Inv$ and every $A \in \calB^+(\Inv)$.
Since every state in $\Inv$ reaches the petite $\Goal$ almost surely (this is the defining property of $\Inv$), we can invoke \citet[Prop.~9.1.7]{meyn2012markov} to conclude that the $\psi$-irreducible chain $(\Inv, T_\Inv)$ is recurrent. 

We now introduce three lemmas about petite sets from \cite{xu2018note}.

\begin{lemma}[Petite Set Properties]\label{lem:petite_properties}
Let $(X, T_X)$ be a $\psi$-irreducible Markov process. Then:
\begin{enumerate}
    \item (Subset) Any measurable subset of a petite set is petite.
    \item (Union) The union of two petite sets is petite.
\end{enumerate}
\end{lemma}

\begin{lemma}[Uniform Hitting Probabilities]\label{lem:hitting_uniform}
Let $(X, T_X)$ be a $\psi$-irreducible Markov process on $(X, \mathcal{B}(X))$ with $\{C_n\}_{n\geq 1} \subset \mathcal{B}(X)$ satisfying $C_n \uparrow X$. For any $\varepsilon > 0$, there exists $D \in \mathcal{B}(X)$ such that:
\begin{enumerate}
    \item $\psi(D) > 1 - \varepsilon$, and
    \item For all $t > 0$, $\lim_{n\to\infty} \mathbb{P}_x(\sigma_{C_n^c} \leq t) = 0$ uniformly on $D$,
\end{enumerate}
where $\sigma_A = \inf\{t \geq 0 : X_t \in A\}$ is the hitting time.
\end{lemma}

\begin{lemma}[Tail Behavior for Recurrent Processes]\label{lem:recurrent_tail}
Let $(X, T_X)$ be a $\psi$-irreducible recurrent Markov process and $C \in \mathcal{B}^+(X)$. For any $\varepsilon > 0$, there exists $F \in \mathcal{B}(X)$ such that:
\begin{enumerate}
    \item $\psi(F) > 1 - \varepsilon$, and
    \item $\lim_{t\to\infty} \mathbb{P}_x(\sigma_C > t) = 0$ uniformly on $F$.
\end{enumerate}
\end{lemma}

Note that \citet{xu2018note} use Egorov's theorem to demonstrate the uniform convergence in \cref{lem:hitting_uniform}, thus eliminating the need for the Weak Feller property used in \cite{meyn2012markov,MajumdarSS24}.

We now present the proof of our main theorem. Our approach proceeds in three stages: first we construct the drift function $V$ through an application of fundamental results from the literature, then we use $V$ to build our variant function $U$.
\rev{The function $V$ so constructed is finite only outside a $\psi$-null set, while \cref{rule:des-ast-markov} asks for a finite-valued supermartingale.
The third stage restricts the invariant to the states from which the infinite values of $V$ are never reached.}
Both are built for the minimal exempt set $T = \Goal$ --- $V$ vanishing exactly on $\Goal$ and satisfying \textbf{V1}, and $U$ satisfying \textbf{V2} --- which suffices for completeness, since $\Goal$ is itself a permissible exempt set of \cref{rule:des-ast-markov}. 

The foundation of our argument relies on combining two important theoretical tools: the general framework from \cite[Th.~ 9.4.2]{meyn2012markov} and the specialized construction in \cite[Th.~1]{xu2018note} applied to our state space $\Inv$.
We produce the essential steps here to maintain self-containment and to establish the specific properties needed for our subsequent construction of $U$.

Our first step establishes the necessary foundations through Lemma~\ref{lem:petite_properties} and Lemma~\ref{lem:hitting_uniform}.
Rather than an arbitrary petite decomposition used by \cite{xu2018note}, we build $V$ on the \emph{same} exhaustion that underlies the variant $U$: an increasing sequence of sets $(C_0, C_1, C_2, \dotsc)$ defined inductively by $C_0 \coloneqq \Goal$ and
\begin{equation}\label{eq:Cn}
C_{n+1} \coloneqq C_n \cup \{x \in \Inv : T_\Inv(x, C_n) > 1/2^n\}, \quad n \in \nats.
\end{equation}
We show below that each $C_n$ is petite and that $\bigcup_{n} C_n = \Inv$, so $(C_n)$ is a nested sequence of petite sets with $C_n \uparrow \Inv$.
\rev{Each $C_n$ is measurable, by induction on $n$: $C_0 = \Goal$ is measurable, and if $C_n$ is measurable, then $x \mapsto T_\Inv(x, C_n)$ is measurable by the definition of a kernel, so $\set{x \in \Inv \mid T_\Inv(x, C_n) > 1/2^n}$ is measurable.}
Since $T_\Inv(x, C_0) > 1/2^0 = 1$ is impossible, $C_1 = C_0 = \Goal$; hence $C_1 = \Goal$ is petite and satisfies the non-triviality condition $\psi(C_1) = \psi(\Goal) > 0$, as $x_0$ reaches $\Goal$ almost surely.
Building $V$ on this exhaustion is what will bound the variant $U$ on the sublevel sets of $V$.

\paragraph{The petite exhaustion $\{C_n\}$.}
We first verify the two properties of the exhaustion \eqref{eq:Cn}: each $C_n$ is petite, and the sequence covers $\Inv$.
Since $C_0 = \Goal$ is petite, \cref{claim:doubling} together with the union property of \cref{lem:petite_properties} gives, by induction on $n$, that every $C_n$ is petite.
\begin{claim}\label{claim:doubling}
If \( E_1 \) is a petite set, then for any real $\beta > 0$, the set
\[
E_2 := \left\{ x \in \Inv : T_\Inv(x, E_1) > \beta \right\}
\]
is also petite.
\end{claim}

\begin{proof}
Since \( E_1 \) is petite, by definition there exist a non-trivial measure \( \nu \) and a probability distribution \( a \) on \( \nats \) such that
\[
K_a(x, A) := \sum_{n=0}^{\infty} a(n) T_\Inv^n(x, A) \geq \nu(A), \quad \forall x \in E_1, \ \forall A \subseteq \Inv.
\]

Define:
\[
\tilde{\nu} := \beta \nu, \quad
\tilde{a}(n) :=
\begin{cases}
a(n+1), & n \geq 0, \\
0, & n = 0.
\end{cases}
\]

We will show that:
\[
K_{\tilde{a}}(x, A) := \sum_{n=0}^{\infty} \tilde{a}(n) T_\Inv^n(x, A) \geq \tilde{\nu}(A), \quad \forall x \in E_2, \ \forall A.
\]

Let \( x \in E_2 \). Then:
\[
\begin{aligned}
K_{\tilde{a}}(x, A) &= \sum_{n=0}^{\infty} \tilde{a}(n) T_\Inv^n(x, A)
= \sum_{n=1}^{\infty} a(n) T_\Inv^n(x, A)
\\&
= \sum_{n=1}^{\infty} a(n) \int_\Inv T_\Inv(x, dy) T_\Inv^{n-1}(y, A)
= \int_\Inv T_\Inv(x, dy) \sum_{n=1}^{\infty} a(n) T_\Inv^{n-1}(y, A) \\
&= \int_{E_1} T_\Inv(x, dy) \sum_{n=1}^{\infty} a(n) T_\Inv^{n-1}(y, A) + \int_{\Inv \setminus E_1} T_\Inv(x, dy) \sum_{n=1}^{\infty} a(n) T_\Inv^{n-1}(y, A).
\end{aligned}
\]

Now observe that for \( y \in E_1 \),
\[
\sum_{n=1}^{\infty} a(n) T_\Inv^{n-1}(y, A) = \sum_{m=0}^{\infty} a(m+1) T_\Inv^m(y, A) = K_{\tilde{a}}(y, A),
\]
so we get:
\[
K_{\tilde{a}}(x, A) \geq \int_{E_1} T_\Inv(x, dy) K_{\tilde{a}}(y, A) \geq \int_{E_1} T_\Inv(x, dy) \nu(A) = \nu(A) \cdot T_\Inv(x, E_1).
\]

By the definition of \( E_2 \), \( T_\Inv(x, E_1) \geq \beta \), then
\[
K_{\tilde{a}}(x, A) \geq \beta \nu(A) = \tilde{\nu}(A).
\]

Hence, \( E_2 \) is petite.
\end{proof}

\begin{claim}\label{claim:cover}
$\bigcup_{n=0}^\infty C_n = \Inv$.
\end{claim}

\begin{proof}
Define the increasing sequence of sets $(B_0, B_1, B_2, \ldots)$ where:
\[ B_n := \{x \in \Inv : \bbP_x(\sigma_\Goal \leq n) > 0\}, \quad n \in \nats. \]
Note that $B_0 = \Goal = C_0$. Since $\bbP_x(\sigma_\Goal < \infty) = 1$ for all $x \in \Inv$, we have $\bigcup_{n=0}^\infty B_n = \Inv$.

We prove by induction that $B_k \subseteq \bigcup_{n=0}^\infty C_n$ for all $k \geq 0$.

\textbf{Base case ($k=0$):} Trivially holds as $B_0 = C_0$.

\textbf{Inductive step:} Assume $B_k \subseteq \bigcup_{n=0}^\infty C_n$. Take any $x \in B_{k+1} \setminus B_k$ with $x \notin \bigcup_{n=0}^\infty C_n$. Then:
\begin{enumerate}
    \item From $x \in B_{k+1} \setminus B_k$: $T_\Inv(x, B_k) > 0$
    \item From $x \notin \bigcup_n C_n$: $T_\Inv(x, C_0) + \sum_{i=0}^{n-1} T_\Inv(x, C_{i+1}\setminus C_i) < \frac{1}{2^n}$ for all $n \in \nats$
\end{enumerate}
This implies $T_\Inv(x, C_0) < \frac{1}{2^n}$ for all $n$, hence $T_\Inv(x, C_0) = 0$, contradicting $T_\Inv(x, B_k) > 0$ since $B_k \subseteq \bigcup_n C_n$.

Therefore, $B_{k+1} \subseteq \bigcup_{n=0}^\infty C_n$, completing the induction.
\end{proof}

\paragraph{Construction of the drift $V$.}
We construct $V$ as a subsequence sum of hitting functions on the exhaustion $(C_n)$.
Fix the reference set $C := C_1 = \Goal$ (petite, with $\psi(C) > 0$ by the above), and for each $n$ let $V_n(x) := \bbP_x(\sigma_{C_n^c} < \sigma_C)$ (\eqref{eq:Vn} below) be the probability of leaving the shell $C_n$ before returning to $C$; each $V_n$ is a supermartingale off $C$ (\eqref{eq:PVleqV}) and equals $1$ outside $C_n$.
A sum $V = \sum_k V_{n_k}$ therefore grows without bound as $x$ leaves every shell, making $V$ unbounded off petite sets --- the coercivity criterion \textbf{V1} needs --- provided the subsequence $\{n_k\}$ is chosen so the sum is also finite $\psi$-almost everywhere.
To make that choice we combine \cref{lem:hitting_uniform} (bounding the probability of escaping $C_n$ early) with \cref{lem:recurrent_tail} (via the recurrence shown above, bounding a late return to $C$), producing sets $A_m$ of $\psi$-measure approaching $1$ on which each $V_{n_k}$ is uniformly small.

To control the convergence behavior, we employ a precision parameter sequence $\{\varepsilon_m\}_{m\geq 1}$ with $\varepsilon_m \searrow 0$. Lemma~\ref{lem:hitting_uniform} guarantees that for each precision level $m$, there exists a control set $L’_m$ with:
\begin{itemize}
    \item Measure guarantee: $\psi(L’_m) > 1 - \varepsilon_m$
    \item Uniform convergence: $\sup_{x\in L’_m} |f_n^{(i)}(x)| \to 0$ as $n\to\infty$
\end{itemize}
where $f_n^{(i)}(x) = \bbP_x(\sigma_{C_n^c} = i)$ represents the hitting probability.

We then construct an optimized sequence of control domains through the telescoping intersection:
\begin{equation}\label{eq:control_sets}
L_1 = L’_1, \quad L_m = \bigcap_{k=1}^m L’_k \text{ for } m \geq 2
\end{equation}
This construction ensures two critical properties:
\begin{enumerate}
    \item Monotonicity: $L_m \subseteq L_{m+1}$ for all $m\geq 1$
    \item Measure control: $\psi(L_m) > 1 - \varepsilon_m$ (by subadditivity)
\end{enumerate}
Moreover, the uniform convergence property is preserved on each $L_m$ for all transition indices $i$.

Let \(\{\varepsilon_m : m \geq 1\}\) be the same precision sequence with \(\varepsilon_m \downarrow 0\).
From \cref{lem:recurrent_tail}, for each \(m\), there exists \(F_m\) such that \(\psi(F_m) > 1 - \varepsilon_m\) and
\[
g_n(x) := \bbP_x(\sigma_C > n) \to 0
\]
uniformly in \(x \in F_m\). Assume \(F_m \uparrow\).

Let  
\[
A_m := L_m \cap F_m, \quad m \geq 1,
\]  
then 
\begin{itemize}
    \item
(a) \(A_m \uparrow\) and \(\psi\left(\bigcup_{m=1}^\infty A_m\right) = 1\).  
\item (b) For all \(i \geq 2\), both
\[
\bbP_x(\sigma_{C_n^c} < i) \text{ and } \bbP_x(\sigma_C > n)
\]
converge to zero uniformly in \(x \in A_m\).
\end{itemize}
Using (a) and (b) the function can be constructed. Define
\begin{equation}\label{eq:Vn}
V_n(x) := \bbP_x(\sigma_{C_n^c} < \sigma_C), \quad n \geq 1.
\end{equation}
\rev{Each $V_n$ is measurable: $V_n(x)$ is the limit as $k \to \infty$ of $\bbP_x(\sigma_{C_n^c} < \sigma_C \text{ and } \sigma_{C_n^c} \leq k)$, which is measurable in $x$ by the argument of \cref{lem:measurable_inv}.}
Then
\[
T_\Inv V_n(x) = \bbP_x(\sigma_{C_n^c} < \sigma_C) = V_n(x), \quad x \in C^c \cap C_n.
\]
Since \(V_n(x) = 1\) on \(C_n^c\), and \(C^c \subset (C^c \cap C_n) \cup C_n^c\), we have
\begin{equation}\label{eq:PVleqV}
T_\Inv V_n(x) \leq V_n(x), \quad x \in C^c.
\end{equation}

We aim to demonstrate that, by appropriately selecting a subsequence \(\{n_k\}\), the function  
\begin{equation}\label{eq:CLF}
V(x) := V_1(x) + \sum_{k=1}^\infty V_{n_k}(x), \quad x \in \Inv,
\end{equation}
satisfies the required properties.
\rev{For every choice of the subsequence, $V$ is measurable, as a pointwise limit of finite sums of the measurable functions $V_n$.}
Specifically, \(V\) meets the drift criterion \textbf{V1}, owing to the linearity of the involved operations.

The leading term $V_1(x) = \bbP_x(\sigma_{C_1^c} < \sigma_C) = \mathbf{1}[x \notin \Goal]$ (recall $C_1 = C = \Goal$) is bounded and is itself an $n=1$ instance of \eqref{eq:Vn}, so it satisfies \eqref{eq:PVleqV} and does not affect the drift property just noted; its purpose is to make $V$ vanish \emph{exactly} on $\Goal$: every summand vanishes on $\Goal$, while $V(x) \geq V_1(x) = 1 > 0$ for $x \notin \Goal$. Thus $V(x) = 0 \Leftrightarrow x \in \Goal$.


Let \(i \geq 2\), then  
\[
V_n(x) \leq \bbP_x(\sigma_{C_n^c} < i) + \bbP_x(\sigma_C > i).
\]

Fix \(k\),  from (b) and choose \(i_k\) such that
\[
\bbP_x(\sigma_C > i_k) < 2^{-(k+1)}, \quad x \in A_k.
\]
Then choose \(n_k\) such that
\[
\bbP_x(\sigma_{C_{n_k}^c} < i_k) < 2^{-(k+1)}, \quad x \in A_k.
\]
So  
\begin{equation}\label{eq:Vnkbound}
V_{n_k}(x) < 2^{-k}, \quad x \in A_k.
\end{equation}

From this:
\[
V(x) = V_1(x) + \sum_{k=1}^\infty V_{n_k}(x) \leq 1 + \sum_{k=1}^m V_{n_k}(x) + \sum_{k=m+1}^\infty 2^{-k} \leq m + 2, \quad x \in A_m.
\]
From (a), hence,
\begin{equation}\label{eq:Vfinite}
\psi(\{x : V(x) = \infty\}) = 0.
\end{equation}
\rev{This is weaker than the finite-valued $V$ that \cref{rule:des-ast-markov} asks for.
We close the gap at the end of the proof, once $U$ is available, because the restriction of the invariant must be shown to preserve criterion \textbf{V2} as well as criterion \textbf{V1}.}

By Fubini's theorem and \eqref{eq:PVleqV}, \(V(x)\) satisfies criterion \textbf{V1}.
To show \(V\) is unbounded off petite sets, note that for any \(m\),
\[
\{x : V(x) \leq m\} \subset C_{n_{m+1}},
\]
a petite set, which completes the construction of $V$\rev{: it satisfies the drift criterion \textbf{V1} on $\Inv$ with values in $[0, \infty]$, and is finite outside the $\psi$-null set of \eqref{eq:Vfinite}}.

\paragraph{Construction of the variant $U$.}
We can now define the variant $U \colon \Inv \to \nnreals$ explicitly:
\[
U(x) =
\begin{cases}
0 & \text{if } x \in \Goal \\
n + 1 & \text{if } x \in C_{n+1}\setminus C_n, n\in \nats
\end{cases}
\]
\rev{$U$ is measurable, being constant on $\Goal$ and on each of the measurable sets $C_{n+1} \setminus C_n$.}
\begin{claim}\label{claim:variant-v2}
The variant $U$ satisfies the variant criterion \textbf{V2}.
\end{claim}

\begin{proof}
By construction $U(x) = 0$ exactly when $x \in C_0 = \Goal$, and, since $U(x) = j$ on $C_j \setminus C_{j-1}$,
\begin{equation}\label{eq:Ulevel}
\{x \in \Inv : U(x) \leq j\} = C_j, \qquad j \in \nats.
\end{equation}
Take the assistant functions $H(r) := n_{\lceil r\rceil+1}$, $d(r) := 1$, and $\varepsilon(r) := \rev{\min\set{1, 2^{1-r}}}$.

\emph{Bound on the sublevel sets of $V$.}
Because $V$ was built on the cover $(C_0, C_1, C_2, \ldots)$, the subsumption $\{x : V(x) \leq m\} \subseteq C_{n_{m+1}}$ for every $m \in \nats$ combines with \eqref{eq:Ulevel} to give, for every real $r \geq 0$,
\[
V(x) \leq r \;\Longrightarrow\; x \in \{V \leq \lceil r\rceil\} \subseteq C_{n_{\lceil r\rceil+1}} = \{U \leq n_{\lceil r\rceil+1}\},
\]
that is, $U(x) \leq H(r)$. As $\{n_k\}$ is increasing, $H$ is non-decreasing.

\emph{Descent.}
Let $x \in C_{n+1}\setminus C_n$, such that $U(x) = n+1$ (in particular $x \notin \Goal$).
By the definition \eqref{eq:Cn} of $C_{n+1}$ we have $T_\Inv(x, C_n) > 1/2^n$, and $x_{k+1} \in C_n$ forces $U(x_{k+1}) \leq n = U(x) - 1$ by \eqref{eq:Ulevel}; hence
\[
T_\Inv\big(x, \{y \in \Inv : U(y) \leq U(x) - d(r)\}\big) \;\geq\; T_\Inv(x, C_n) \;>\; \frac{1}{2^n}.
\]
\rev{If in addition $U(x) \leq r$, then $n + 1 \leq r$, and this probability exceeds $2^{-n} \geq 2^{1-r} \geq \varepsilon(r)$.}
\rev{For $r < 1$ there is no $x \in \Inv \setminus \Goal$ with $U(x) \leq r$, and the descent condition is vacuous.}
With $d(r) = 1$, $U$ meets \textbf{V2}.
\end{proof}

\begin{revblock}
\paragraph{Making $V$ finite.}
\cref{rule:des-ast-markov} asks for a finite-valued $V$, while \eqref{eq:Vfinite} guarantees finiteness only $\psi$-almost everywhere.
We extend the level-set notation to the value $\infty$, so that $V_{=\infty} = \{x \in \Inv \mid V(x) = \infty\}$ and $V_{<\infty} = \Inv \setminus V_{=\infty}$. 
We show that the set $V_{<\infty}$ contains $x_0$ and $\Goal$ and satisfies the invariant condition, making it a suitable invariant under which $V$ is finite.

Clearly, $x_0 \in V_{<\infty}$ and $\Goal \subseteq V_{<\infty}$ because $V(\Goal) = 0$.
Every goal state moves to $x_0$ in one step, and $x_0$ reaches $\Goal$ almost surely, so the state at time $\sigma_\Goal + 1$ is distributed as $\delta_{x_0}$.
Hence
\[
\mathbf{1}[x_0 \in V_{=\infty}] = \bbP_{x_0}\left(x_{\sigma_\Goal + 1} \in V_{=\infty}\right) \leq \sum_{n \geq 1} T_\Inv^n(x_0, V_{=\infty}) = 0,
\]
where the inequality holds because $\sigma_\Goal + 1$ is almost surely a finite index at least $1$, and the final equality holds because every term of $\psi$ is non-negative and $\psi(V_{=\infty}) = 0$ by \eqref{eq:Vfinite}.

Additionally, $V_{<\infty}$ satisfies the invariant condition of \cref{rule:des-ast-markov}: at every $x \in V_{<\infty} \setminus \Goal$, the kernels $T_X$ and $T_1$ coincide, $T_1(x, \Inv) = 1$, and \eqref{eq:PVleqV} with Fubini's theorem gives $\int_\Inv T_\Inv(x, dy)\, V(y) \leq V(x) < \infty$, so $T_\Inv(x, V_{=\infty}) = 0$; hence $T_X(x, V_{<\infty}) = 1$.

It remains to restrict the certificates to $V_{<\infty}$.
Take the invariant $V_{<\infty}$ and the exempt set $T := \Goal$.
The rule asks for certificates defined on all of $X$, while the drift $V$ of \eqref{eq:CLF} and the variant $U$ of \cref{claim:variant-v2} are defined on $\Inv$ only, and $V$ is infinite on $V_{=\infty}$.
Take the functions $V', U' : X \to \nnreals$ that agree with $V$ and $U$ on $V_{<\infty}$ and take the value $1$ at every state outside $V_{<\infty}$, together with the assistant functions $H$, $d$, and $\varepsilon$ of \cref{claim:variant-v2}.
$V'$ and $U'$ are measurable because $V$, $U$, and $V_{<\infty}$ are.

It is easy to see that these functions satisfy the conditions of \cref{rule:des-ast-markov}; we make this more explicit now.
Outside the invariant, the rule constrains a certificate only through the two conditions $V'(x) = 0 \Leftrightarrow x \in T$ and $U'(x) = 0 \Leftrightarrow x \in \Goal$, and the value $1$ meets both because $T = \Goal \subseteq V_{<\infty}$.
Writing $V$ and $U$ for $V'$ and $U'$ on $V_{<\infty}$, the set $V_{<\infty}$ contains $x_0$ and $\Goal$ and satisfies the invariant condition of the rule.
$V$ is finite-valued on $V_{<\infty}$, and there $V(x) = 0$ precisely when $x \in T$ and $U(x) = 0$ precisely when $x \in \Goal$, as on $\Inv$.
The drift condition of criterion \textbf{V1} holds at every $x \in V_{<\infty} \setminus T$, because $\int_{V_{<\infty}} T_X(x, dy)\, V(y) \leq \int_\Inv T_X(x, dy)\, V(y) \leq V(x)$.
The condition $V(x) \leq r \Rightarrow U(x) \leq H(r)$ of criterion \textbf{V2} holds at every $x \in V_{<\infty}$, because it holds at every $x \in \Inv$.
The descent condition of criterion \textbf{V2} holds at every $x \in V_{<\infty} \setminus \Goal$ with $U(x) \leq r$, because $T_X(x, \Inv \setminus V_{<\infty}) = 0$ gives
\[
T_X\left(x, \set{y \in V_{<\infty} \mid U(y) \leq U(x) - d(r)}\right) = T_X\left(x, \set{y \in \Inv \mid U(y) \leq U(x) - d(r)}\right) \geq \varepsilon(r).
\]
This completes the proof of \cref{th:ast-markov}.
\end{revblock}

%


\newoutputstream{todos}
\openoutputfile{main.todos.ctr}{todos}
\addtostream{todos}{\arabic{@todonotes@numberoftodonotes}}
\closeoutputstream{todos}

\label{endofdocument}
\newoutputstream{pagestotal}
\openoutputfile{main.pagestotal.ctr}{pagestotal}
\addtostream{pagestotal}{\getpagerefnumber{endofdocument}}
\closeoutputstream{pagestotal}

\end{document}